\documentclass[10pt]{article}
\usepackage{amsfonts}

\usepackage{amsmath,amssymb}
\usepackage{graphicx,color}
\usepackage{epstopdf}
\usepackage{epsfig}
\usepackage{mathrsfs}
\usepackage[breaklinks=true]{hyperref}

\newcommand{\beq}{\begin{equation}}
\newcommand{\eeq}{\end{equation}}
\newcommand{\bqa}{\begin{eqnarray}}
\newcommand{\eqa}{\end{eqnarray}}

\definecolor{green}{rgb}{0.00,0.50,0.00}

\newtheorem{theorem}{Theorem}

\newtheorem{corollary}[theorem]{Corollary}

\newtheorem{definition}[theorem]{Definition}
\newtheorem{example}[theorem]{Example}

\newtheorem{lemma}[theorem]{Lemma}

\newtheorem{proposition}[theorem]{Proposition}
\newtheorem{remark}[theorem]{Remark}

\newenvironment{proof}[1][Proof]{\noindent\textbf{#1.} }{\ \rule{0.5em}{0.5em}}

\begin{document}

\title{Entropy Production and Information Flow for Markov Diffusions with Filtering
}


\author{John E. Gough\thanks{Aberystwyth University, SY23 3BZ, Wales, United Kingdom, \texttt{jug@aber.ac.uk}}         \and
        Nina H. Amini \thanks{Laboratoire des Signaux et Syst\'{e}mes, CNRS, 91192 Gif sur Yvette, France,  
				\texttt{nina.amini@lss.supelec.fr}}}

\maketitle

\begin{abstract}
Filtering theory gives an explicit models for the flow of information and thereby quantifies the rates of change of information supplied to and dissipated from the filter's memory. Here we extend the analysis of Mitter and Newton \cite{MN05} from linear Gaussian models to general nonlinear filters involving Markov diffusions.The rates of entropy production are now generally the average squared-field (co-metric) of various logarithmic probability densities, which may be interpreted as Fisher information associate with Gaussian perturbations (via de Bruijn's identity). We show that the central connection is made through the Mayer-Wolf and Zakai Theorem for the rate of change of the mutual information between the filtered state and the observation history. In particular, we extend this Theorem to cover a Markov diffusion controlled by observations process, which may be interpreted as the filter acting as a Maxwell's D\ae mon applying feedback to the system.
\end{abstract}

\section{Introduction}
The aim of the present paper is to extend the analysis of Mitter and Newton \cite{MN05} on the information flow in linear Kalman-Bucy filters to the general setting of nonlinear filters for Markov diffusion models. This leads us into the geometric setting of stochastic processes and filtering theory. In particular, the diffusion tensor defines a co-metric - introduced into Probability Theory by P.-A.  Meyer and also known as the $\Gamma$-operator, or l'op\'{e}rature carr\'{e} du champ, see \cite{Emery}, or \cite{Ledoux}. The co-metric quantifies the extent to which the generator of a diffusion process differs from a tangent vector field, and to a certain extent describes the irreversibility of the stochastic dynamics. Building on the theory of nonlinear filtering \cite{DavisMarcus}-\cite{Zak69}, we use classical results of Kadota, Zakai and Ziv \cite{KZZ71}, and of Mayer-Wolf and Zakai \cite{MWZ} to compute the rate of change of mutual information shared between an unknown state $X(t)$ and the observations $Y_0^t$, see also \cite{Zakai05}, \cite{WKP13}.

We consider a system whose state, $X(t)\in \mathbb{R}^n$, undergoes a stochastic dynamical evolution with probability density of $X(t)$ denoted as 
\begin{eqnarray}
\rho_t (x) \equiv e^{- \Phi_t (x)}.
\end{eqnarray}
(More generally, we shall refer to $\Phi (x) = - \ln \rho (x)$ as being the \textit{surprise potential} associated with a probability density $\rho $
for a random variable $X$. We understand that $\Phi (x)$ takes the value $+\infty$ whenever $\rho (x) =0$. The average surprise potential is then the Shannon entropy $\mathsf{H} (X) = \mathbb{E} [ \Phi (X)]
\equiv -\int \rho (x) \ln \rho (x) \, dx$.)

The uncertainty in the state may be measured by $\mathsf{H}_t = -\int \rho_t \ln \rho_t$ and should be monotonically increasing with time. In principle, we may make observations, $Y(t)$, which depend on the current state and its past history, see Figure \ref{fig:state_observations}. 
 
\begin{figure}[h]
	\centering
		\includegraphics[width=0.50\textwidth]{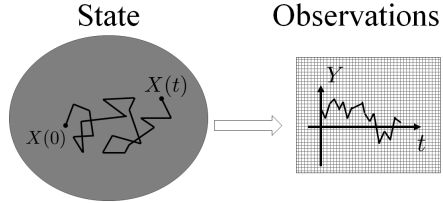}
	\caption{A randomly evolving state, $X(t)$, which is being monitored leading to partial observations $Y(t)$.}
	\label{fig:state_observations}
\end{figure}

The observations may be only partial, and furthermore subject to noise. The observer may be passive, but may also alter the evolution of the system based on the observed data - that is using feedback to act as a controller, or Maxwell's d\ae mon. 
At this stage we may flip from the pejorative view of Shannon entropy as a measure of uncertainty, to the more benign one of being a measure of information. The inverse problem of trying to guess the state $X(t)$ form the observations $Y_0^t =\{ Y(s) : 0 \le s \le t \}$ is a staple of signal processing: while usually not possible to completely determine the state, $X(t)$, one aims instead to construct a filter which computes the optimal estimate based on the observations, see Figure \ref{fig:observations_filter}. 

\begin{figure}
	\centering
		\includegraphics[width=0.50\textwidth]{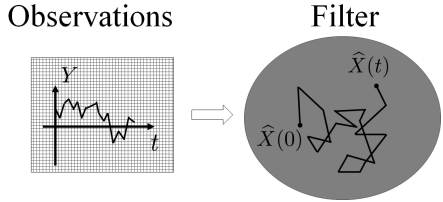}
	\caption{The estimated the state, $\widehat{X}(t)$, based on the observational data, $Y_0^t$.}
	\label{fig:observations_filter}
\end{figure}

In this paper, we will consider the standard problem where the state $X(t)$ undergoes a Markov diffusion and the filter computes an estimate which is optimal in the least squares sense. 
For an arbitrary bounded measurable function $f$, the goal is then to calculate the conditional expectation
\begin{eqnarray}
\pi_t (f)  \triangleq \mathbb{E} \big[ f(X(t)) | \mathcal{Y}_0^t \big] ,
\end{eqnarray}
where $\mathcal{Y}_0^t$ be the $\sigma$-algebra generated by the family $Y_0^t$. $\pi_t (f)$ is then the least-squares estimate of $f(X(t))$ given the observations. The problem is mathematically equivalent to computing the probability density valued process, $\widehat{\rho}_t$ \cite{LipsterShiryaev}, adapted to the filtration $\mathcal{Y} = \{ \mathcal{Y}_0^t :t \geq 0 \}$ such that
\begin{eqnarray}
\pi_t (f) \equiv \int f(x) \widehat{\rho}_t (x , \omega ) \ dx .
\end{eqnarray}
Intuitively, the flow of information to the filter should lead to a reduction in the uncertainty in the state. However, this is quantifiable.

One may further consider the effect of feedback, see \cite{DelvenneSandberg} for discrete time and \cite{AKW14} for continuous time models.

\subsection{Linear Gaussian Models}
In this subsection we recall the results of Mitter and Newton \cite{MN05} concerning the entropy production and information flow of the Kalman-Bucy filter. In this model, the processes $X$ and $Y$ are jointly Gaussian and the filter can be computed exactly.

\subsubsection{Entropy Production}
We consider a linear quadratic Gaussian (LQG) model where we may compute the
various quantities explicitly. Here we take the SDE for the diffusion to be 
\begin{eqnarray}
dX\left( t\right) =AX\left( t\right) dt+BdW\left( t\right)
\label{eq:SDE_L}
\end{eqnarray}
where $A\in \mathbb{R}^{n\times n}$ and $B\in \mathbb{R}^{n\times m}$, with $%
W$ being a canonical $m$-dimensional Wiener process. The initial condition
is $X\left( 0\right) \sim \mathcal{N}\left( \mu _{0},V_{0}\right) $ and one finds that 
$X\left( t\right) \sim \mathcal{N} \left( \mu _t,V\left( t\right) \right) $ where the
(unconditioned) mean value of $X(t)$ is 
\begin{eqnarray}
\mu_t = e^{-At } \, \mathbb{E} [ X(0) ] ,  \label{eq:mu_t}
\end{eqnarray}
and this decays to zero if $A$ is Hurwitz.

The \textit{diffusion matrix} associated with problem is given by $\Sigma \equiv BB^{\top }$, which is constant in both space and time.
The covariance matrix at
time $t$ satisfies the differential equation 
\begin{eqnarray}
\frac{d}{dt}V\left( t\right) =AV\left( t\right) +V\left( t\right) A^{\top
}+\Sigma ,  \label{eq:rate_V}
\end{eqnarray}
with $\Sigma =BB^{\top }$ which we assume to be invertible. Under conditions
that $A$ is Hurwitz, we find that there exists a steady state, $\rho_{\text{ss}}$ which is a $\mathcal{N}\left( \mu _{0},V_{%
\text{ss}}\right) $ distribution where the steady state covariance satisfies 
\begin{eqnarray}
AV_{\text{ss}}+V_{\text{ss}}A^{\top }+\Sigma =0.  \label{eq:V_ss}
\end{eqnarray}
We note that $\frac{d}{dt}V\left( t\right) =A\left( V\left( t\right) -V_{%
\text{ss}}\right) +\left( V\left( t\right) -V_{\text{ss}}\right) A^{\top }$.

The steady state surprise potential is then 
\begin{eqnarray}
\Phi _{\text{ss}}\left( x\right) \equiv \frac{1}{2}x^{\top }V_{\text{ss}%
}^{-1}x+E_{0}
\end{eqnarray}
where the constant is $E_{0}=\ln \left\{ \sqrt{2\pi }^{n}\left| V_{\text{ss}%
}\right| \right\} $.

At this stage, Mitter and Newton \cite{MN05} introduce a quantity 
\begin{eqnarray}
\mathsf{E}_t \triangleq \mathbb{E} [ \Phi_{\text{ss}} (X(t)) ],
\end{eqnarray}
which plays a role similar to the free energy in thermodynamics. We will refer to it as
the \textit{internal surprise} at time $t$, and in the LQG case it is 
\begin{eqnarray}
\mathsf{E}_{t}\equiv \frac{1}{2}\int \rho _{t}\left( x\right) \text{tr}
\left\{ V_{\text{ss}}^{-1}xx^{\top } \right\} +E_{0}=\frac{1}{2}\text{tr}%
\left\{ V_{\text{ss}}^{-1}V\left( t\right) \right\} +E_{0}.
\end{eqnarray}
We see that 
\begin{eqnarray}
\frac{d}{dt}\mathsf{E}_{t} &\equiv &\frac{1}{2}\text{tr}\left\{ V_{\text{ss}%
}^{-1}\left[ A\left( V\left( t\right) -V_{\text{ss}}\right) +\left( V\left(
t\right) -V_{\text{ss}}\right) A^{\top }\right] \right\}  \notag \\
&=&\text{tr}\left\{ V_{\text{ss}}^{-1}A\left( V\left( t\right) -V_{\text{ss}%
}\right) \right\} .
\end{eqnarray}

The Shannon entropy for Gaussian distributions is well known, and here we have
\begin{eqnarray}
\mathsf{H}_{t}=\frac{1}{2}\ln \left| V\left( t\right) \right| +\frac{n}{2}%
\ln \big( 2\pi e \big).
\end{eqnarray}
and we compute its rate of change using the following Lemma.

\begin{lemma}
\label{lem:rate_log_det} Suppose that the matrix valued function $V$
satisfies the linear differential equation $\frac{d}{dt}V\left( t\right)
=AV\left( t\right) +V\left( t\right) A^{\top }+\Sigma $, (\ref{eq:rate_V}),
then 
\begin{eqnarray}
\frac{d}{dt} \ln | V(t) | = \text{tr} \big\{ 2 A + V(t)^{-1} \Sigma \big\}.
\end{eqnarray}
\end{lemma}

\begin{proof}
We have the well-known identity $\frac{d}{dt} \ln | V(t) | = \text{tr} %
\big\{ V(t)^{-1} \frac{dV(t)}{dt} \big\}$, which in this case implies $\frac{%
d}{dt} \ln | V(t) | = \text{tr} \big\{ A+A^\top +V(t)^{-1} \Sigma \big\}$,
where we use the commutativity under the trace. The result then follows from
noting that $\text{tr} \big\{ A \big\} = \text{tr} \big\{ A^\top \big\}$.
\end{proof}

\bigskip

We therefore have that 
\begin{eqnarray}
\frac{d}{dt}\mathsf{H}_{t}=\text{tr}\left\{ V\left( t\right) ^{-1}A\left(
V\left( t\right) + \frac{1}{2} \Sigma \right) \right\} \equiv \text{tr}%
\left\{ V\left( t\right) ^{-1}A\left( V\left( t\right) -V_{\text{ss}}\right)
\right\} ,
\end{eqnarray}
where, for the last step, we substitute (\ref{eq:V_ss}) and once again we
use the invariance of the trace under commutation and transposition.

Mitter and Newton \cite{MN05} then observe that the quantity analogous to the free energy,
$\mathsf{F}_t \triangleq \mathsf{E}_t - \mathsf{H}_t$ is non-increasing.
The rate of change of the $\mathsf{F}_t$, which we refer to as the \textit{free surprise}, is given explicitly
by
\begin{eqnarray}
\frac{d}{dt}\mathsf{F}_{t} &\equiv &\text{tr}\left\{ \left[ V_{\text{ss}%
}^{-1}-V\left( t\right) ^{-1}\right] A\left( V\left( t\right) -V_{\text{ss}%
}\right) \right\} \\
&=&-\text{tr}\left\{ \left[ V_{\text{ss}}^{-1}-V\left( t\right) ^{-1}\right]
AV_{\text{ss}}\left[ V_{\text{ss}}^{-1}-V\left( t\right) ^{-1}\right]
V\left( t\right) \right\} \\
&=&-\frac{1}{2}\text{tr}\left\{ \left[ V_{\text{ss}}^{-1}-V\left( t\right)
^{-1}\right] \Sigma \left[ V_{\text{ss}}^{-1}-V\left( t\right) ^{-1}\right]
V\left( t\right) \right\}
\end{eqnarray}
from which we see that $\frac{d}{dt}\mathsf{F}_{t}\leq 0$ with equality if and only if the covariance matrix equals the steady state value.

\subsubsection{The Kalman-Bucy Filter}
Mitter and Newton \cite{MN05} extend their analysis of entropy and information flow by considering the Kalman-Bucy filter.
Here the state dynamics will be just the LQG model already presented in (\ref{eq:SDE_L}), but now take into account the information supplied by observations $Y\left( t\right) \in \mathbb{R}^{p}$ which are likewise assumed
to be linear with 
\begin{eqnarray}
dY\left( t\right) =CX\left( t\right) dt+dU\left( t\right) ,
\label{eq:KB_obs}
\end{eqnarray}
with $C\in \mathbb{R}^{p\times n}$ and the observational noise $U$ is taken to be independent of
dynamical noise $W$ for simplicity. 
The initial variables $X(0)$ and $Y(0)$ are both assumed to be independent
of the noises $W$ and $U$, and jointly Gaussian. The pair of processes $(X,Y)$ are then jointly Gaussian.

The conditioned probability density turns out to be again Gaussian and is  given as
\begin{eqnarray}
\widehat{\rho}_t (x, \omega ) = \frac{1}{ \sqrt{2 \pi}^n  | \widehat{V}_t |} e^{- \frac{1}{2} [x- \widehat{X}_t (\omega ) ]^\top
\widehat{V}_t^{-1} [x- \widehat{X}_t (\omega )]},
\end{eqnarray}
where the mean $\widehat{X}_t (\omega )$ and covariance matrix $\widehat{V}(t)$ are explicitly computed as follows:
the covariance matrix is deterministic and satisfies the Riccati
equation 
\begin{eqnarray}
\frac{d}{dt}\widehat{V}_{t}=A\widehat{V}_{t}+\widehat{V}_{t}A^{\top
}+BB^{\top }-\widehat{V}_{t}C^{\top }C\widehat{V}_{t},  \label{eq:KB2}
\end{eqnarray}
with $\widehat{V}_{0}=Cov\left( X\left( 0\right) \right) $;
the mean is $\widehat{X}_{t}=\pi _{t}\left( X\right) $ and satisfies
\begin{eqnarray}
d\widehat{X}_{t}=A\widehat{X}_{t}dt+\widehat{V}_{t}C^{\top }dI\left( t\right),
\label{eq:KB1}
\end{eqnarray}
where $I$ is a stochastic process known as the innovation process 
\begin{eqnarray}
dY\left( t\right)
=CX(t)+dU(t)\equiv C\widehat{X}_{t}\,dt+dI\left( t\right) .
\end{eqnarray}
Note that we then have 
\begin{eqnarray*}
dI\left( t\right) =C\varepsilon _{t}\left( X\right) dt+dU\left( t\right) ,
\end{eqnarray*}
where the state error is $\varepsilon _{t}\left( X\right) =X\left( t\right) -%
\widehat{X}_{t}$. It is not too difficult to show that the state error
satisfies 
\begin{eqnarray*}
d\varepsilon _{t}\left( X\right) =A\varepsilon _{t}\left( X\right) dt+BdV\left(
t\right) -\widehat{V}\left( t\right) C^{\top }dI\left( t\right) .
\end{eqnarray*}

From the point of view of filtering theory, the natural quantity to look at  is the mutual information shared between the state 
and the observations history, $\mathsf{I}\left( X\left( t\right) ;Y_{0}^{t}\right)$. 
Mitter and Newton \cite{MN05} observe the decomposition
\begin{eqnarray}
\frac{d}{dt}\mathsf{I}\left( X\left( t\right) ;Y_{0}^{t}\right) = \dot{%
\mathsf{S}} (t)- \dot{\mathsf{D} } (t) .
\end{eqnarray}
where the \textbf{information supplied to the filter's memory storage} up to
time $t$ is 
\begin{eqnarray}
\mathsf{S} (t) \triangleq \mathsf{I}\left( X_0^t;Y_{0}^{t}\right) +\mathsf{I}%
\left( X(0) ;Y(0) \right)
\end{eqnarray}
and $\mathsf{D} (t) $ is interpreted as the \textbf{information dissipated by the filter's memory storage} up to
time $t$.

For the Kalman-Bucy filter, they show that
\begin{eqnarray}
\frac{d}{dt}\mathsf{S}_{\text{K-B}} (t) &=& \frac{1}{2} \text{tr} \big\{ C \widehat{V} (t) C^\top \big\} ,\nonumber \\
 \frac{d}{dt} \mathsf{D}_{\text{K-B}} (t) &=& \frac{1}{2}   \text{tr} \big\{ \Sigma \big( \widehat{V}(t)^{-1} -V(t)^{-1} \big) \big\}.
\end{eqnarray}
The rate of change of $\mathsf{S}_{\text{K-B}} (t)$ follows from a well known result of signal process due to Duncan \cite{TD70}, 
see Remark \ref{rem:KB_Duncan}. We will re-derive these in Example \ref{example:KB_rates}.

\section{Notation and Background Concepts}
We recall for completeness the fundamental information theoretic and probabilistic concepts needed in this paper.
\subsection{Measures of Entropy}
Suppose that we have a continuous $\mathbb{R}^{n}$-valued random variable $X$
possessing a probability distribution function $\rho _{X}$, then we define the \textbf{surprise} function $\Phi_X$ by
\begin{eqnarray}
\rho_X (x) \equiv e^{- \Phi_X (x)}.
\end{eqnarray}
We shall use the terms entropy and information interchangeably.
\subsubsection{Shannon Entropy}
The Shannon entropy of $X$ is defined to be the average surprise:
\begin{eqnarray}
\mathsf{H}\left( X\right) \triangleq \mathbb{E} [ \Phi_X ] \equiv -\int \rho _{X}\ln \rho _{X}.
\end{eqnarray}
Strictly speaking the entropy is a functional of the probability distribution, $\rho _{X}$, rather than $X$ itself, but we will use this
standard short hand convention throughout. If we have a family $(X_1 , \cdots , X_n )$ of random variables on the same probability space, then we will of course understand $H(X_1 , \cdots , X_n )$ to be their entropy of their joint distribution.
\subsubsection{Relative Entropy}
Let $\rho _{1}$ and $\rho _{2}$ be probability density functions on $\mathbb{R}^{n}$, then the \textit{relative entropy} of $\rho _{1}$ with respect to $\rho _{2}$ is defined to be 
\begin{eqnarray}
\mathsf{D}\left( \rho _{1}||\rho _{2}\right) \triangleq \int \rho _{1}\ln \frac{\rho
_{1}}{\rho _{2}} 
\equiv \mathbb{E}_1 [ \Phi_2 - \Phi_1 ].
\end{eqnarray}
This is also known as the \textit{Kullback-Leiber divergence}, the \textit{entropy of discrimination} and the \textit{information distance}.
We note the Gibbs' inequality $\mathsf{D}\left( \rho _{1}||\rho _{2}\right) \geq 0$ with equality if and only if $\rho
_{1}=\rho _{2}$.

\subsubsection{Mutual Entropy}
Given a pair of random variables $X$ and $Y$ with joint density $\rho_{X,Y}$ and marginals $\rho_X, \, \rho_Y$ respectively, the \textit{mutual information} of $X$ and $Y$ is defined to be 
\begin{eqnarray}
\mathsf{I}\left( X;Y\right) \triangleq \mathsf{I}\left( Y;X\right) =\mathsf{D}\left( \rho _{X,Y}||\rho _{X}\rho
_{Y}\right) .
\end{eqnarray}
In detail, we have 
\begin{eqnarray*}
\mathsf{I}\left( X;Y\right) =\int \rho _{X,Y}\ln \left[ \frac{\rho _{X,Y}}{\rho
_{X}\rho _{Y}}\right]\equiv  \mathsf{H}\left( X\right) +\mathsf{H}\left( Y\right) -\mathsf{H}\left( X,Y\right)   .
\end{eqnarray*}
The mutual information is the difference from the actual joint distribution to
the product of the marginals. In this sense, it is a measure of the degree
of correlation. As an immediate consequence of the Gibbs' inequality, we
have that the mutual information is non-negative definite, i.e. $\mathsf{I}\left(
X,Y\right) \geq 0$ and we have inequality if and only if $X$ and $Y$ are
independent.

\subsubsection{Conditional Entropy}
The \textbf{conditional information} of a
random variable $X$ given $Y$ is defined to be 
\begin{eqnarray}
\mathsf{H}\left( X|Y\right) \triangleq -\int \rho _{X,Y}\ln \left[ \frac{\rho _{X,Y}}{%
\rho _{Y}}\right] 
\mathsf{H}\left( X|Y\right) =\mathsf{H}\left( X,Y\right) -\mathsf{H}\left( Y\right)
\label{eq:cond_entropy+id}
\end{eqnarray}
If $X$ and $Y$ are independent, then $\mathsf{H}\left( X|Y\right) =\mathsf{H}\left( X\right) $
and $\mathsf{H}(Y|X)=\mathsf{H}(Y)$, and in particular, $\mathsf{I}(X;Y ) \equiv 0$.

The mutual information satisfies the following relation: 
\begin{eqnarray}
\mathsf{I}\left( X;Y\right) &=&\mathsf{H}\left( X\right) -\mathsf{H}\left( X|Y\right) .
\label{eq:mutual_cond}  
\end{eqnarray}

\subsubsection{Fisher Entropy}
A final form of entropy that will be of relevance to use is the Fisher information, specifically the translation Fisher information we recall below.
Let $\{ \rho _{X}^{\theta } : \theta \in \Theta \}$ be a family of probability density functions for a random variable $X$
parametrized by a parameter taking values in an open subset $\Theta$ of $\mathbb{R}^p$. We say that $\theta \mapsto \rho _{X}^{\theta }\left(
x\right) $ is the \textit{likelihood} of $\theta $ given the observation $x$.
\index{likelihood} 
The \textit{score} is defined to be the random vector in $\mathbb{R}^{p}$
with components 
\index{score}
\begin{eqnarray}
s_{i}^{\theta }\left( X\right) \triangleq \frac{\partial }{\partial \theta
_{i}}\ln \rho _{X}^{\theta }.
\end{eqnarray}
It is easy to see that $\mathbb{E}_{\theta }\left[ s_{i}^{\theta }\left(
X\right) \right] \equiv 0$.
The \textit{Fisher information} is the $p\times p$ matrix with entries 
\begin{eqnarray}
\mathsf{F}_{ij}^{\theta }\left( X\right) \triangleq \mathbb{E}_{\theta }%
\left[ s_{i}^{\theta }\left( X\right) s_{j}^{\theta }\left( X\right) \right]
.
\end{eqnarray}
In other words, the Fisher information is the covariance matrix for the
score variable. In detail, we have 
\begin{eqnarray*}
\mathsf{F}_{ij}^{\theta }\left( X\right) =\int \rho _{X}^{\theta }\left( 
\frac{\partial }{\partial \theta _{i}}\ln \rho _{X}^{\theta }\right) \left( 
\frac{\partial }{\partial \theta _{j}}\ln \rho _{X}^{\theta }\right) .
\end{eqnarray*}

Let us suppose that we know that the distribution of a random $n$-vector is
known up to translation. That is, we have a fixed distribution $\rho $ and
the actual distribution is of the form 
\begin{eqnarray*}
\rho ^{\theta }\left( x\right) =\rho \left( x-\theta \right)
\end{eqnarray*}
where $\theta \in \mathbb{R}^{n}$ is unknown parameter. 
In this case the
score is $s^{\theta }\left( X\right) =\nabla _{\theta }\ln \rho \left(
x-\theta \right) $ but this can be equivalently written as $-\nabla _{x}\ln
\rho \left( x-\theta \right) $. 
We shall refer to the associated Fisher information matrix in this case as the \textit{translational Fisher information}, and denote it by $\mathsf{J}(X)$. It takes the
form 
\begin{eqnarray*}
\mathsf{F}_{ij}^{\theta }\left( X\right) =\mathsf{J}_{ij} \left( X\right) \triangleq\int \rho \left( x-\theta \right)
s_{i}^{\theta }\left( X\right) s_{j}^{\theta }\left( X\right) \,dx
\end{eqnarray*}
however we can perform the volume preserving change of variable $x^{\prime
}=x-\theta $ to get the simplified form 
\begin{eqnarray}
\mathsf{J}_{ij}\left( X\right) &=&\int \rho \frac{\partial \ln \rho }{%
\partial x_{i}}\frac{\partial \ln \rho }{\partial x_{j}} .
\label{eq:Fisher_J}
\end{eqnarray}
We also have $\mathsf{J}_{ij}\left( X\right) \equiv \int \frac{1}{\rho }\frac{\partial \rho }{\partial x_{i}}\frac{%
\partial \rho }{\partial x_{j}} \equiv  -\int \rho \frac{\partial ^{2}}{%
\partial x _{i}\partial x _{j}}\ln \rho $. Note that this Fisher information does not depend on the value $\theta $,
only on the representative distribution $\rho $. We also note the rescaling law
$ \mathsf{J}\left( aX+b\right) =\frac{1}{a^{2}}\mathsf{J}\left( X\right) $.
The well-known Cram\'{e}r-Rao inequality is equivalent to $Cov\left( X\right) \geq  \mathsf{J}\left( X\right) ^{-1}$.

An important connection between Shannon and Fisher information arises in the context of Gaussian perturbations.
We have, for instance, \textit{De Bruijn's Identity} which states that if $X$ is a random $n$-vector and $Z$ be a standard
Gaussian $n$-vector independent of $X$, then 
\begin{eqnarray}
  \frac{d}{dt}\mathsf{H}\left( X+\sqrt{t}Z\right)  =\frac{1}{2}%
\mathrm{tr}\mathsf{J}\left( X + \sqrt{t} Z \right) .
\label{eq:De_Bruijn}
\end{eqnarray}

We also note the log-Sobolev inequality for Fisher information. Let $X$ be a random $n$-vector and suppose that $\mathsf{J}\left( X+\sqrt{t}Z\right) $ exists and is differentiable in $t$
in a neighborhood of 0 for arbitrary standard Gaussian $n$-vector
perturbations $Z$ independent of $X$. Then 
\begin{eqnarray}
\left. \frac{d}{dt}\frac{n}{\mathsf{J}\left( X+\sqrt{t}Z\right) }\right|
_{t=0}\geq 1.
\label{eq:log_Sob}
\end{eqnarray}

\subsection{Stochastic Dilations}

We fix a probability space $\left( \Omega ,\mathcal{A},\mathbb{P}\right) $.
Let $\mathcal{F}=\left\{ \mathcal{F}_{0}^{t}\right\} _{t\geq 0}$ be a
filtration on a sample space $\Omega $, that is, an increasing and
right-continuous family of $\sigma $-subalgebras of $\mathcal{A}$, with each 
$\mathcal{F}_{0}^{t}$ containing all sets of measure zero. As usual, we say
that a process $X=\left\{ X\left( t\right) \right\} _{t\geq 0}$ is $\mathcal{%
F}$-adapted if $X\left( t\right) $ is $\mathcal{F}_{0}^{t}$-measurable for
each $t\geq 0$. We say that $M$ is a martingale wrt. the filtration if it is
an integrable process such that $\mathbb{E}\left[ M\left( t\right) |\mathcal{%
F}_{0}^{s}\right] \equiv M\left( s\right) $ for all $0\leq s\leq t$.

More generally, if $X$ is a stochastic process, we define its \textit{%
conditional derivative wrt. the filtration} as 
\begin{eqnarray}
\mathbb{D}_{\mathcal{F}}X\left( t\right) \triangleq \lim_{\tau \rightarrow
0^{+}}\mathbb{E} [ \frac{X\left( t+\tau \right) -X\left( t\right) }{\tau }%
\mid \mathcal{F}_{0}^{t} ] .
\end{eqnarray}
We may refer to $\mathbb{D}_{\mathcal{F}}X\left( t\right) $ as the
conditional velocity wrt. the filtration. Martingales are then processes
with zero conditional velocity (or, equivalently, vanishing trend). For a
pair of processes $X$ and $Y$, we define their \textit{angular bracket wrt.
the filtration} as 
\begin{eqnarray}
\left\langle X,Y\right\rangle _{t,\mathcal{F}}\triangleq \int_{0}^{t}%
\mathbb{E}\left[ dX\left( s\right) dY\left( s\right) |\mathcal{F}_{0}^{s}%
\right] ds.
\end{eqnarray}

Given an operator $\mathcal{L}$ from the smooth functions $\mathscr{C}%
^{\infty }(\mathbb{R}^{n})$ to the bounded continuous function on $\mathbb{R}%
^{n}$, then an $\mathbb{R}^{n}$-valued process, $X$, is said to solve the
martingale problem associated with $\mathcal{L}$ if, for each $f\in %
\mathscr{C}^{\infty }(\mathbb{R}^{n})$, the equation 
\begin{eqnarray}
f\left( X\left( t\right) \right) =f\left( X\left( 0\right) \right)
+\int_{0}^{t}\mathcal{L}f\,\left( X\left( s\right) \right) ds+M_{f}\left(
t\right)
\end{eqnarray}
\textit{defines} a martingale process $M_{f}$. 

In the following we will be interested in with diffusion processes on $%
\mathbb{R}^{n}$, where the associated operator $\mathcal{L}$ is elliptic.
Here we have 
\begin{eqnarray}
\mathbb{D}_{\mathcal{F}}f(X\left( t\right) )\equiv \mathcal{L} f (X\left(
t\right) )
\end{eqnarray}
and the angular bracket takes the form 
\begin{eqnarray}
\frac{d}{dt}\left\langle f\left( X\right) ,g\left( X\right) \right\rangle
_{t,\mathcal{F}}\equiv \left. 2\Gamma _{\mathcal{L}}\left( f,g\right)
\right| _{X\left( t\right) }
\end{eqnarray}
where the \textit{co-metric} is defined to be 
\begin{eqnarray}
\Gamma _{\mathcal{L}}\left( f,g\right) \triangleq \mathcal{L}\left(
fg\right) -\left( \mathcal{L}f\right) g-f\left( \mathcal{L}g\right) .
\end{eqnarray}
The co-metric quantifies the obstruction to $\mathcal{L}$ being a tangent
vector field since $\Gamma _{\mathcal{L}}\left( f,g\right) \equiv 0$ for $%
\mathcal{L}=v^{i}\left( x\right) \partial _{i}$. It has the following
properties:

\begin{enumerate}
\item  $\Gamma _{\mathcal{L}+\mathcal{K}}= \Gamma _{\mathcal{L}}+\Gamma _{%
\mathcal{K}}$;

\item  it is symmetric $\Gamma _{\mathcal{L}}\left( f,g\right) =\Gamma _{%
\mathcal{L}}\left( g,f\right) $;

\item  off-diagonal terms are determined through polarization: 
\begin{eqnarray}
2\Gamma _{\mathcal{L}}\left( f,g\right) =\Gamma _{\mathcal{L}}\left(
f+g,f+g\right) -\Gamma _{\mathcal{L}}\left( f-g,f-g\right) .
\end{eqnarray}
\end{enumerate}

The following is easily verified.

\begin{proposition}
Let $\mathcal{K}f=\alpha f+\beta ^{i}f_{,i}+\frac{1}{2}\gamma ^{ij}f_{,ij}$,
for functions $\alpha ,\beta ^{i}$ and $\gamma ^{ij}=\gamma ^{ji}$, then 
\begin{eqnarray}
\Gamma _{\mathcal{K}}\left( f,g\right) =-\alpha fg+\gamma ^{ij}f_{,i}g_{,j}.
\end{eqnarray}
\end{proposition}

We remark that there is an analogue of the Ricci Tensor  due to Bakry-Emery tensor $%
\Gamma _{\mathcal{L}}^{2}\left( f,g\right) \triangleq \mathcal{L}\left(
\Gamma _{\mathcal{L}}\left( f,g\right) \right) -\Gamma _{\mathcal{L}}\left( 
\mathcal{L}f,g\right) \ -\Gamma _{\mathcal{L}}\left( f,\mathcal{L}g\right) $, see for instance \cite{Ledoux}.
We should remark that Bakry and Emery have developed a theory on inequalities based on their $\Gamma^2$ tensor
which generalize the log-Sobolev inequality (\ref{eq:log_Sob}).

We note that we have $\frac{d}{dt}\mathbb{E}\left[ f\left( X\left( t\right)
\right) \right] =\mathbb{E}\left[ \mathcal{L}f\,\left( X\left( t\right)
\right) \right] $. If define the probability density $\rho _{t}$ for $%
X\left( t\right) $, so that 
\begin{eqnarray}
\mathbb{E}\left[ f\left( X\left( t\right) \right) \right] \equiv \int
f\left( x\right) \rho _{t}\left( x\right) dx
\end{eqnarray}
then the density satisfies the Fokker-Planck equation 
\begin{eqnarray}
\frac{\partial }{\partial t}\rho _{t}=\mathcal{L}^{\star }\rho _{t}
\label{eq:Fokker_Planck}
\end{eqnarray}
where the adjoint operator is defined by the duality 
\begin{eqnarray}
\int \mathcal{L}f\left( x\right) \rho \left( x\right) dx=\int f\left(
x\right) \mathcal{L}^{\star }\rho \left( x\right) dx.
\end{eqnarray}

\section{Entropy Production For Markov Diffusions}

\subsection{Markov Diffusions}

Let us specialize to a state process, $X$, which is diffusion on $\mathbb{R}%
^{n}$ satisfying a stochastic differential equation of the form 
\begin{eqnarray}
dX^{i}\left( t\right) =v^{i}\left( X\left( t\right) \right) dt+\sum_{\alpha
=1}^{r}B_{\alpha }^{i}\left( X\left( t\right) \right) dW^{\alpha }\left(
t\right)  \label{eq:SDE}
\end{eqnarray}
where $W=\left( W^{1},\cdots ,W^{r}\right) $ is an $r$-dimensional Wiener
process (the dynamical noise). The associated operator is then 
\begin{eqnarray}
\mathcal{L}=v^{i}\left( x\right) \partial _{i}+\frac{1}{2}\Sigma ^{ij}\left(
x\right) \partial _{ij},
\end{eqnarray}
where the diffusion tensor field is 
\begin{eqnarray}
\Sigma ^{ij}\left( x\right) =\sum_{\alpha }B_{\alpha }^{i}\left( x\right)
B_{\alpha }^{j}\left( x\right) .
\end{eqnarray}
The diffusion tensor, $\Sigma $, then turns out to be the co-metric tensor as we
have 
\begin{eqnarray}
\Gamma _{\mathcal{L}}\left( f,g\right) \equiv \sum_{\alpha }B_{\alpha
}\left( f\right) B_{\alpha }\left( g\right) =\Sigma ^{ij}f_{,i}g_{,j}=\left(
\nabla f\right) ^{\top }\Sigma \left( \nabla g\right)
\end{eqnarray}
where we have the $r$ tangent vector fields $B_{\alpha }$ with $B_{\alpha
}\left( f\right) =B_{\alpha }^{i}\left( x\right) \,\left( \partial
_{i}f\right) $. The adjoint generator is now given by 
\begin{eqnarray}
\mathcal{L}^{\star }\rho \equiv -\sum_{i}\partial _{i}\left( \rho
v^{i}\right) +\frac{1}{2}\sum_{i,j}\partial _{ij}^{2}\left( \Sigma ^{ij}\rho
\right) .
\end{eqnarray}

A general feature of Markov diffusions is that we now have the additional
properties:

\begin{enumerate}
\item  it is a bi-derivation: $\Gamma _{\mathcal{L}}\left(
f,gh\right) =\Gamma _{\mathcal{L}}\left( f,g\right) h+g\Gamma _{\mathcal{L}%
}\left( f,h\right) $;

\item  it is positive semi-definite: $\Gamma _{\mathcal{L}}\left( f,f\right)
\geq 0$.;
\end{enumerate}

\begin{lemma}
The Fokker-Planck equation is equivalent to the continuity equation 
\begin{eqnarray}
\frac{\partial \rho _{t}}{\partial t}+ \nabla . J_t =0
\end{eqnarray}
where the probability flux density $J_t$ has the components 
\begin{eqnarray}
J^{i}\left( x,t\right) =v^{i}\left( x\right) \rho _{t}\left( x\right) -\frac{%
1}{2}\frac{\partial }{\partial x^{j}}\left( \Sigma ^{ij}\left( x\right) \rho
_{t}\left( x\right) \right) .
\end{eqnarray}
\end{lemma}

We remark that the flux may also be written in the form 
\begin{eqnarray}
J_t = \rho_t u - \frac{1}{2} \Sigma \nabla \rho_t ,
\end{eqnarray}
where we introduce a new velocity field $u$ having the components 
\begin{eqnarray}
u^{i} \triangleq v^{i}-\frac{1}{2}\frac{\partial \Sigma ^{ij}}{\partial x^{j}%
}.
\end{eqnarray}

Let us remark that the dual generator may be written in terms of the new
velocity $u$ as 
\begin{eqnarray}
\mathcal{L}^{\star }f=-\nabla .\left( uf\right) +\frac{1}{2}\Sigma
^{ij}f_{,ij},
\end{eqnarray}
and in particular 
\begin{eqnarray}
\mathcal{L}^{\star }1=- (\nabla .u) .
\end{eqnarray}

If, moreover, $\rho_t (x)$ is strictly positive for all $x$ and $t\ge 0$,
then we may write the flux as 
\begin{eqnarray}
J_t =\rho _{t} \left[ u -\frac{1}{2}\Sigma \,\left( \nabla \ln \rho _{t}
\right) \right] .  \label{eq:J_u1}
\end{eqnarray}
In this case we can write $\rho_t (x) \equiv e^{- \Phi_t (x)}$, where $%
\Phi_t $ may be referred to as the \textbf{surprise potential}. As $\Phi_t =
- \ln \rho_t$, we see that 
\begin{eqnarray}
J_t \equiv \rho _{t}\left[ u -\frac{1}{2}\Sigma \,\left( \nabla \ln \rho
_{t} \right) \right] = \rho _{t}\left[ u +\frac{1}{2}\Sigma \,\left( \nabla
\Phi _{t} \right) \right]  \label{eq:J_u}
\end{eqnarray}

\begin{definition}
We say that the Markov diffusion admits a \textbf{steady state} it there is
probability density, $\rho _{\text{ss}}\left( x\right) $, such that $%
\mathcal{L}^{\star }\rho _{\text{ss}}\equiv 0$.
\end{definition}

If a steady state exists, then the Fokker-Planck equation tells us that it
is invariant and we see that the flux should vanish for this state leaving 
\begin{eqnarray}
u\left( x\right) =\frac{1}{2}\Sigma \left( x\right) \,\nabla \ln \rho _{%
\text{ss}}\left( x\right) = - \frac{1}{2}\Sigma \left( x\right) \,\nabla
\Phi _{\text{ss}}\left( x\right).  \label{eq:u_Phi_ss}
\end{eqnarray}
If we further assume that the steady state is strictly positive then we may
likewise introduce a \textbf{steady state surprise potential} function $\Phi
_{\text{ss}}\left( x\right) =-\ln \rho _{\text{ss}}\left( x\right) $. Note
that we may write 
\begin{eqnarray}
J_t &=& \frac{1}{2} \rho_t \Sigma \, \nabla \left( \Phi_t - \Phi_{\text{ss}}
\right) = -\frac{1}{2} \rho _{t}\Sigma \,\left( \nabla \ln \frac{\rho _{t}}{%
\rho _{\text{ss}}}\right) .
\end{eqnarray}

\subsection{Entropy Production for General Diffusions}

The entropy may be interpreted as the uncertainty, and for diffusions one
would expect the entropy associated with the diffusing variable $X(t)$ to
increase as $t$ increases. Not least because we are continuously injecting
more noise into the dynamics. Indeed, we have a continuous Gaussian
perturbation. We will compute the rate of change of the entropy $\mathsf{H_t}%
=\mathsf{H}\left( X(t)\right) $ associated with the probability density $%
\rho _{t}$ of the diffusion (\ref{eq:SDE}) in this section. We begin,
however, with the deterministic case.

Suppose that we have a deterministic flow generated by a velocity vector
field $v=v^{i}\left( x\right) \partial _{i}$. The dynamics is then given by
the system of ODEs $\frac{d}{dt}X^{i}\left( t\right) =v^{i}\left( X\left(
t\right) \right) $. In this case, the Fokker-Planck equation reduces to the
continuity equation with the simplified flux $J_{t}\left( x\right) =v\left(
x\right) \rho _{t}\left( x\right) $. We then have 
\begin{eqnarray}
\frac{d}{dt}\mathsf{H}\left( X\left( t\right) \right) &=&-\frac{d}{dt}\int
\rho _{t}\ln \rho _{t}=-\int \left( 1+\ln \rho _{t}\right) \frac{\partial
\rho _{t}}{\partial t}=\int \left( 1+\ln \rho _{t}\right) \nabla \left( \rho
_{t}v\right)  \notag \\
&=&-\int \left( \nabla \ln \rho _{t}\right) ^{\top }\rho _{t}v=-\int \left(
\nabla \rho _{t}\right) ^{\top }v=\int \rho _{t}\nabla .v.
\end{eqnarray}
In other words, 
\begin{eqnarray}
\frac{d}{dt}\mathsf{H}\left( X\left( t\right) \right) =\mathbb{E}\left[
\left. \left( \nabla .v\right) \right| _{X\left( t\right) }\right] .
\end{eqnarray}
The rate of production of entropy is therefore the average of the divergence
of the velocity field. This makes intuitive sense as $\nabla .v$ measures
the local change of volume, so its average gives a global measure of change
in uncertainty. In the special case of an incompressible fluid $\left(
\nabla .v\equiv 0\right) $ we have zero entropy production - this includes
Hamiltonian systems by virtue of Liouville's Theorem.

Now let us treat the diffusive case.

\bigskip

\begin{lemma}
\label{eq:L_logf} Let $f$ be a strictly positive twice continuously
differentiable function, then for a Markov diffusion generator, $\mathcal{L}$%
, 
\begin{eqnarray}
\mathcal{L}\ln f=\frac{1}{f}\mathcal{L}f-\frac{1}{2}\Gamma _{\mathcal{L}%
}\left( \ln f,\ln f\right) .
\end{eqnarray}
\end{lemma}

\begin{proof}
This is actually a straightforward derivation: 
\begin{eqnarray}
\mathcal{L}\ln f &=&\frac{1}{f}v^{\top }\nabla f+\frac{1}{2}\Sigma
^{ij}\left( \frac{1}{f}f_{,ij}-\frac{1}{f^{2}}f_{,i}f_{,j}\right) \\
&=&\frac{1}{f}\mathcal{L}f-\frac{1}{2f^{2}}\Gamma _{\mathcal{L}}\left(
f,f\right) = \frac{1}{f}\mathcal{L}f-\frac{1}{2}\Gamma _{\mathcal{L}}\left(
\ln f,\ln f\right) .
\end{eqnarray}
\end{proof}

\begin{theorem}
\label{thm:dot_H} The entropy production rate for a diffusion process $%
X\left( t\right) $ is given by 
\begin{eqnarray}
\frac{d}{dt}\mathsf{H}\left( X\left( t\right) \right) =\mathbb{E}\left[
\left. \left( \nabla .u\right) \right| _{X\left( t\right) }+\left. \frac{1}{2%
}\Gamma _{\mathcal{L}}\left( \ln \rho _{t},\ln \rho _{t}\right) \right|
_{X\left( t\right) } \right] .  \label{eq:lem_id}
\end{eqnarray}
\end{theorem}

\begin{proof}
We now have 
\begin{eqnarray*}
\frac{d}{dt}\mathsf{H}\left( X\left( t\right) \right) &=&-\int \left( 1+\ln
\rho _{t}\right) \frac{\partial \rho _{t}}{\partial t}\\
&&=-\int \left( 1+\ln
\rho _{t}\right) \mathcal{L}^{\star }\rho _{t}=-\int \left( \mathcal{L}\ln
\rho _{t}\right) \rho _{t}
\end{eqnarray*}
but, from Lemma \ref{eq:L_logf}, $\mathcal{L}\ln \rho _{t} =\frac{1}{\rho
_{t}}\mathcal{L}\rho _{t}-\frac{1}{2}\Gamma _{\mathcal{L}}\left( \ln \rho
_{t},\ln \rho _{t}\right) $.

Therefore, 
\begin{eqnarray*}
\frac{d}{dt}\mathsf{H}\left( X\left( t\right) \right) =\int \left( -\mathcal{%
L}\rho _{t}+\frac{1}{2}\rho _{t}\Gamma _{\mathcal{L}}\left( \ln \rho
_{t},\ln \rho _{t}\right) \right) .
\end{eqnarray*}
The first term here is $-\int \left( \mathcal{L}\rho _{t}\right) =-\int \rho
_{t}\mathcal{L}^{\star }1=\int \rho _{t}\nabla .u$, and this observation
leads to the desired result.
\end{proof}

\bigskip

Note that this says that the entropy rate consists of a geometric dilation
term (similar to the deterministic case, but now with the
divergence of $u$ instead of $v$) and an additional term which is irreversible in the sense
that it comes from the second-order terms of the generator.

We remark that Theorem \ref{thm:dot_H} may be proved more directly from the
Fokker-Planck equation and an integration by parts, 
\begin{eqnarray*}
\frac{d}{dt}\mathsf{H}_{t} &=&\int \left( 1+\ln \rho _{t}\right) \left(
\nabla .J\right) \\
&=&-\int \left( \nabla \ln \rho _{t}\right) ^{\top }J \\
&=& -\int \rho_t u^{\top }\left( \nabla \ln \rho _{t}\right) +\frac{1}{2}%
\int \rho_t\left( \nabla \ln \rho _{t}\right) ^{\top }\Sigma \left( \nabla
\ln \rho _{t}\right)
\end{eqnarray*}
and the last step follows directly by substituting in (\ref{eq:J_u}) to
eliminate $J_t$. The last term is readily shown to be $\int \rho_t \left(
\nabla . u \right) $ by a simple integration by parts. This proof is very
much related to the textbook proof of the de Bruijn identity, equation (\ref{eq:De_Bruijn}), where the
Gaussian perturbation is viewed as convolution wrt. the fundamental solution to the heat equation (the Wiener kernel!)
and similarly employing the integration by parts step.

Another point of interest is that the term in (\ref{eq:lem_id})
involving the $\Gamma$-operator of the generator with the score (logarithmic
derivative) as argument is evidently a form of Fisher information. In fact, it reduces
to $\text{tr} \, \mathsf{J}\left( X(t)\right) $ in the special case where $%
\Sigma $ is the identity matrix. One can, in fact, show that if the diffusion matrix $\Sigma $ is invertible and $\rho _{t}$ is
strictly positive, then the equation (\ref{eq:lem_id}) may be written as 
\begin{eqnarray}
\frac{d}{dt}\mathsf{H}_{t}=2\int \frac{1}{\rho _{t}}J_t^{\top }\Sigma
^{-1}J_t-2\int u^{\top }\Sigma ^{-1}J_t .  \label{eq:lem_id1}
\end{eqnarray} 

We can now give the appropriate extension of the Mitter-Newton entropy production equations for Markov diffusions
with a steady state.

\begin{definition}
Suppose that a Markov diffusion $X(t)$ admits a strictly positive, steady
state $\rho _{\text{ss}} \equiv e^{-\Phi _{\text{ss}}}$ The \textbf{free
surprise} associated with the diffusion is defined to be the relative
entropy 
\begin{eqnarray}
\mathsf{F}_{t}\triangleq \mathsf{D}\left( \rho _{t}||\rho _{\text{ss}%
}\right) .
\end{eqnarray}
The \textbf{internal surprise} associated with the diffusion $X(t)$ is
defined to be 
\begin{eqnarray}
\mathsf{E}_{t}\triangleq \mathbb{E}\left[ \Phi \left( X(t)\right) \right] .
\end{eqnarray}
\end{definition}

We note that 
\begin{eqnarray}
\mathsf{F}_{t} =\int \rho _{t}\ln \frac{\rho _{t}}{\rho _{\text{ss}}} =
\int \rho _{t}\ln \rho _{t}-\int \rho _{t}\ln \rho _{\text{ss}} ,
\end{eqnarray}
and so the free surprise, internal surprise and entropy, associated with a
diffusion $X(t)$ with steady state $\rho _{\text{ss}}$, are related by 
\begin{eqnarray}
\mathsf{F}_{t}=\mathsf{E}_{t}-\mathsf{H}_{t}.  \label{eq:F=E-H}
\end{eqnarray}
We again have the clear analogy with thermodynamics that was the central motivation of Mitter and Newton
\cite{MN05}. The free and internal
surprises correspond to the free and internal energies, while the
information is the entropy - the relation (\ref{eq:F=E-H}) is then the
standard thermodynamic relation with unit temperature. The free surprise is non-increasing and its rate of change
involves the co-metric.

\begin{theorem}
\label{thm:unobserved_F} Let $\left\{ X\left( t\right) \right\} _{t\geq 0}$
be the diffusion process satisfying the stochastic differential equation (%
\ref{eq:SDE}) and assume that the diffusion matrix $\Sigma $ is everywhere
invertible and that there exists a unique steady state $\rho _{\text{ss}}$.
Then the associated free surprise is non-increasing, and we have 
\begin{eqnarray}
\frac{d}{dt}\mathsf{F}_{t}=-\frac{1}{2}\mathbb{E}\left[ \left. \Gamma _{%
\mathcal{L}}\left( \ln \frac{\rho _{t}}{\rho _{\text{ss}}},\ln \frac{\rho
_{t}}{\rho _{\text{ss}}}\right) \right| _{X\left( t\right) }\right] \leq 0.
\label{eq:dot_F}
\end{eqnarray}
\end{theorem}

\begin{proof}
We have 
\begin{eqnarray*}
\mathsf{F}_{t} &=&\int (1+\ln \rho _{t}-\ln \rho _{\text{ss}})\frac{\partial
\rho _{t}}{\partial t} \\
&=&\int (1+\ln \rho _{t}-\ln \rho _{\text{ss}})\mathcal{L}^{\star }\rho _{t}
\\
&=&\int \left( \mathcal{L}\ln \frac{\rho _{t}}{\rho _{\text{ss}}}\right)
\rho _{t} \\
&=&\int \left[ \frac{\rho _{ss}}{\rho _{t}}\left( \mathcal{L}\frac{\rho _{t}%
}{\rho _{\text{ss}}}\right) -\frac{1}{2}\Gamma _{\mathcal{L}}\left( \ln 
\frac{\rho _{t}}{\rho _{\text{ss}}},\ln \frac{\rho _{t}}{\rho _{\text{ss}}}%
\right) \right] \rho _{t}
\end{eqnarray*}
where we used Lemma \ref{eq:L_logf}. The first of these terms vanishes as it
equals 
\begin{eqnarray*}
\int \rho _{\text{ss}}\left( \mathcal{L}\frac{\rho _{t}}{\rho _{\text{ss}}}%
\right) =\int \left( \mathcal{L}^{\star }\rho _{\text{ss}}\right) \frac{\rho
_{t}}{\rho _{\text{ss}}}=0,
\end{eqnarray*}
since $\mathcal{L}^{\star }\rho _{\text{ss}}\equiv 0$ for a steady state.
\end{proof}

\bigskip

We remark that if $\rho_t$ is strictly positive, and if $\Sigma$ is
everywhere invertible, then the (\ref{eq:dot_F}) is equivalent to 
\begin{eqnarray}
\frac{d}{dt}\mathsf{F}_{t}=-\frac{1}{2}\int \frac{1}{\rho _{t}}J_{t}^{\top
}\Sigma ^{-1}J_{t}\leq 0.
\end{eqnarray}
This may be alternatively proved by noting that $\frac{d}{dt}\mathsf{F}_{t}=%
\frac{d}{dt}\mathsf{E}_{t}-\frac{d}{dt}\mathsf{H}_{t}$. The rate of change
of the internal surprise is 
\begin{eqnarray}
\frac{d}{dt}\mathsf{E}_{t} &=&\frac{d}{dt}\int \rho _{t}\left( x\right) \Phi
_{\text{ss}}\left( x\right) =\int \frac{\partial \rho _{t}\left( x\right) }{%
\partial t}\Phi _{\text{ss}}\left( x\right) \\
&=&-\int \left( \nabla .J_{t}\right) ^{\top }\Phi _{\text{ss}} =\int
J_{t}^{\top }\nabla \Phi _{\text{ss}} =-2\int J_{t}^{\top }\Sigma u,
\end{eqnarray}
where we used the Fokker-Planck equation, integration by parts, and for the
last step (\ref{eq:u_Phi_ss}). The term $\frac{d}{dt}\mathsf{E}_{t}$
therefore corresponds to the second term in (\ref{eq:lem_id1}) which now
reads as $\frac{d}{dt}\mathsf{H}_{t}=\frac{1}{2}\int \frac{1}{\rho _{t}}%
J_{t}^{\top }\Sigma ^{-1}J_{t}+\frac{d}{dt}\mathsf{E}_{t}$, and this gives
the desired result.

\bigskip



\section{Filtering \& Entropy Reduction}

The filtering problem considered here is the following standard one. We have
a diffusion process, $X$, on $\mathbb{R}^{n}$ and make observations leading
to an $\mathbb{R}^{p}$-valued process $Y$. Denoting the $\sigma $-algebra
generated by the observations over the time interval $\left[ 0,t\right] $ by 
$\mathcal{Y}_{0}^{t}$, we wish to compute the filtered estimate 
$\pi _{t}\left( f\right) \triangleq \mathbb{E}\left[ f\left( X\left( t\right)
\right) |\mathcal{Y}_{0}^{t}\right] $.
Equivalently, we may compute the $\mathcal{Y}$-adapted probability density
valued process $\widehat{\rho}$ such that 
$\pi _{t}\left( f\right) \equiv \int f\left( x\right) \widehat{\rho}_{t}\left(
x\right) dx$. 
We refer to $\widehat{\rho}_{t}$ as the \textit{conditioned density} and recall that $%
\mathbb{E}\left[ \widehat{\rho}_{t}\right] \equiv \rho _{t}$.

\bigskip

Let us suppose that the observations satisfy 
\begin{eqnarray}
dY\left( t\right) =H\left( t\right) dt+dU\left( t\right)
\end{eqnarray}
where $H$ is an $\mathbb{R}^{p}$-valued process adapted to the filtration
generated by $X$, and $U$ is a $p$-dimensional Wiener process independent of 
$X$. We encounter the drift term $H\left( t\right) \equiv \mathbb{D}_{%
\mathcal{U}}H\left( t\right) $, where $\mathcal{U}$ is the filtration
generated by the observational noise $U$. However, as we only observe $Y$ it
makes more sense to introduce the new process 
\begin{eqnarray}
\widehat{H}\left( t\right) \triangleq \mathbb{D}_{\mathcal{Y}}Y\left( t\right) .
\end{eqnarray}
A Theorem of Duncan \cite{TD70} states that the mutual information shared between the
signal $H_{0}^{t}$ and the observations $Y_{0}^{i}$ is given by 
\begin{eqnarray}
\mathsf{I}\left( H_{0}^{t};Y_{0}^{t}\right) =\frac{1}{2}\int_{0}^{t}\mathbb{E%
}\left[ \left| \varepsilon \left( H,s\right) \right| ^{2}ds\right]
\end{eqnarray}
where the estimation error is 
\begin{eqnarray}
\varepsilon \left( H,t\right) =H\left( t\right) -\widehat{H}\left( t\right) ,
\end{eqnarray}
or equivalently, $\left( \mathbb{D}_{\mathcal{U}}-\mathbb{D}_{\mathcal{Y}%
}\right) Y\left( t\right) $. The proof makes extensive use of Girsanov transformations \cite{Girsanov}.

\begin{remark}
\label{rem:KB_Duncan}
Duncan's Theorem takes a simple form for the Kalman-Bucy filter. As $H(t) = CX(t)$, we see that $\varepsilon (H,t) \equiv C \varepsilon_t (X)$ where $\varepsilon_t (X)$ is the state error.
Therefore, $\mathbb{E}\left[ \left| e(H, \tau  ) \right| ^{2}\right] = \text{tr} \, \left\{ C^{\top }C \widehat{V} ( \tau ) \right\} $and therefore the mutual information is
\begin{eqnarray}
\mathsf{I}_{\text{K-B}} \left( Y_{0}^{t};X_{0}^{t}\right) =\frac{1}{2}\int_{0}^t\text{tr}\left\{
C \widehat{V} \left( \tau \right) C^{\top }\right\} d\tau .
\label{eq:I_KB}
\end{eqnarray}
\end{remark}


We may also write
\begin{eqnarray}
dY\left( t\right) \equiv \widehat{H} (t)\,dt+dI\left( t\right)
\end{eqnarray}
where the \textit{innovations process} $\left\{ I\left( t\right) \right\}
_{t\geq 0}$ is defined as 
\begin{eqnarray}
dI\left( t\right) &\triangleq &dY\left( t\right) -\widehat{H} (t) \,dt 
\notag \\
&\equiv & dU \left( t\right) +\varepsilon \left( H,t\right) \, dt.  
\label{eq:dI-dW}
\end{eqnarray}
The increment $dI\left( t\right) $ is the difference between the observed
signal, $dY\left( t\right) $, and what we would have expected, $\widehat{H}%
(t) \, dt$, given the observations up to time.

The innovations process is a martingale wrt. the filtration $%
\left\{ \mathcal{Y}_{0}^t \right\} _{t\geq 0}$, and furthermore we have $%
\left( dI\right) ^{2}=\left( dU\right) ^{2}=dt$ so by L\'{e}vy's Theorem it
is a Wiener process with respect to this filtration.  

\subsection{Filtering Markov Diffusions}

We consider the problem of a system with unobserved state $X(t)$ in $\mathbb{%
R}^n$. We assume that it evolves according to the stochastic differential
equation (\ref{eq:SDE}) which we write in the vector form 
\begin{eqnarray}
(\text{State Dynamics}) \qquad dX(t) = v\left( X(t)\right) dt+ B \left(
X(t)\right) dW(t) ,  \label{eq:filter_dynamics}
\end{eqnarray}
with $W$ an $r$-dimensional canonical Wiener process, as before.

Information about the state comes from an observation process in $\mathbb{R}%
^p$ 
\begin{eqnarray}
(\text{Observations)} \qquad \qquad \qquad dY (t) = h\left( X(t)\right)
dt+dU (t) ,  \label{eq:filter_obs}
\end{eqnarray}
or $dY^i (t) = h^i\left( X(t)\right) dt+dU^i (t) $ in component form. Again we assume that $U$ is a $p$-dimensional Wiener process. The dynamical noise $W$ is assumed to be statistically independent of the observational noise $U$ corrupting observed signals. We shall denote by $\mathcal{Y}_0^t$ the $\sigma$-algebra generated by the observations  $\{ Y(s) : 0 \le s \le t\}$. 
To initialize the problem, we may fix a probability measure $\mathbb{P}%
_{X(0), Y(0)}$ for the initial state and observation.

The filter for the problem (\ref{eq:filter_dynamics}, \ref{eq:filter_obs}) can be written as
\begin{eqnarray}
\pi _{t}\left( f\right) \equiv \frac{\sigma
_{t}\left( f\right) }{\sigma _{t}\left( 1\right) }
\label{eq:filter_pi_ratio}
\end{eqnarray}
where $\sigma_t (f)$ as the \textit{unnormalized filter}, or the 
\textit{Zakai filter}. The normalization $\sigma_t (1)$ will be a stochastic
process adapted to the filtration of $Y$.   
The process $\sigma_t (f)$ satisfies the \textit{Zakai equation}
\begin{eqnarray}
d\sigma _{t}\left( f\right)  &=&\sigma _{t}\left( \mathcal{L}f\right)
dt+\sigma _{t} ( fh^{\top } ) \,dY (t) ,
\label{eq:Zakai}
\end{eqnarray}
where $\mathcal{L}$ is the generator of the diffusion $X$.
The normalization process $\sigma _{t}\left( 1\right) $ is given by
\begin{eqnarray}
\sigma _{t}\left( 1\right) =\exp \left\{ \int_{0}^{t}\pi _{s}\left(
h\right) ^{\top }dY\left( s \right) -\frac{1}{2}\int_{0}^{t}
| \pi _{s}\left( h\right) |^2 ds\right\},
\label{eq:sigma_1}
\end{eqnarray}
and, in particular,
\begin{eqnarray}
\mathbb{E}\left[ \ln \sigma _{t}\left( 1\right) \right] =\frac{1}{2}%
\int_{0}^{t}\mathbb{E}\left[  |\pi _{s}\left( h\right) |^2  \right] ds .  \label{q:sigma_1_average}
\end{eqnarray}
  
\bigskip

From the Zakai filter, we may deduce the form of the normalized filter $%
\pi_t(f)$.

\begin{theorem}[Kushner-Stratonovich]
\index{Theorem!Kushner-Stratonovich}
The filter $\pi_t (f)$ satisfies the  
equation
\begin{eqnarray}
d\pi _{t} ( f )  = \pi _{t}\left( \mathcal{L}f\right) dt +\left[
\pi _{t} ( fh^{\top } ) - \pi _{t}\left( f\right) \pi _{t} (
h^{\top } ) \right]  \,dI(t),
\label{eq:Kushner_Stratonovich}
\end{eqnarray}
known as the \textbf{Kushner-Stratonovich} equation, 
where $\{ I(t) : t \ge 0\} $ is the \textbf{innovations process}: 
\begin{eqnarray}
dI(t)\triangleq
dY (t)-\pi _{t}\left( h\right) dt.
\end{eqnarray}
\end{theorem}
  
For completeness, we give a derivation of the Kushner-Stratonovich equation, (\ref{eq:Kushner_Stratonovich}),
in in Section \ref{app:A}.
We say that the filtering problem admits a \textbf{conditional density} $\widehat{\rho}_t $ if we have
\begin{eqnarray}
\widehat{f}_t ( \omega) =\int f(x) \widehat{\rho}_t (x, \omega )dx.
\end{eqnarray}
In this case the Kushner-Stratonovich equation is equivalent to 
\begin{eqnarray}  \label{eq:KS_density}
d\widehat{\rho}_t (x, \omega ) &=&\mathcal{L}^{\star} \widehat{\rho}_t (x,\omega ) \,
dt \nonumber \\
&&+ \widehat{\rho}_t (x, \omega ) \, \left\{ h^\top (x) -\int h
(x^{\prime})^\top \widehat{\rho} _t (x^{\prime}, \omega )dx^{\prime}\right\} \,
dI(t , \omega ).  \notag \\
\end{eqnarray}
Note that this is \emph{nonlinear} in $\widehat{\rho}_t$.

If we average over all outputs we get an average density $\rho_t$ which is
just the unconditional density for the diffusion $X$, and which satisfies
the Fokker-Planck (Kolmogorov forward) equation, $\frac{d}{dt} \rho _t =\mathcal{L}^\star \rho_t$.

The Zakai filter can similarly be reformulated so that 
\begin{eqnarray}
\sigma_t (f) \equiv \int f(x) \zeta_t (x, \omega ) dx
\end{eqnarray}
for an \emph{non-normalized} stochastic density function $\zeta _{t}$. The
Zakai equation (\ref{eq:Zakai}) is then 
\begin{eqnarray}
d\zeta _{t}(x, \omega ) =\mathcal{L}^\star \zeta _{t} (x, \omega ) \,dt+
\zeta_{t}(x, \omega ) \, h(x)^\top dY (t, \omega ).  \label{eq:Zakai_z}
\end{eqnarray}
Note that this is a \textit{linear} equation for the density $\zeta_t$.

We then have 
\begin{eqnarray*}
\widehat{\rho} _{t} (x , \omega ) = \frac{\zeta _{t}(x , \omega ) }{\int \zeta
_{t}\left( x^{\prime}, \omega \right) dx^{\prime}}.
\end{eqnarray*}
This follows from the It\={o} calculus with the assumption that the
innovations constitute a standard Wiener process. Recall that the
normalization factor is $\int \zeta_t (x, \omega ) dx \equiv \sigma_t (1)$.

\begin{definition}
The \textbf{state error} is defined to be $\varepsilon _{t}\left( X\right) =X\left( t\right) -%
\pi_{t}(X)$, and the \textbf{conditioned covariance matrix} is defined as $\widehat{V}\left(
t\right) =\mathbb{E}\left[ | \varepsilon _{t}\left( X\right) \varepsilon _{t}\left(
X\right) ^{\top }\right] $. 
\end{definition}

The conditioned covariance matrix may also be written as 
\begin{eqnarray}
\widehat{V}_{t}\equiv \pi _{t}\left( XX^{\top }\right) -\pi_t (X) \pi_t (X)
^{\top }.
\end{eqnarray}

More generally, we define the least squares error in estimating $f(X(t))$
from the observations to be 
\begin{eqnarray}
\varepsilon_{t} (f) \triangleq f(X(t))- \mathbb{E} [ f(X(t)) | \mathcal{Y}_{t]}
] =f(X(t)) -\pi_t (f) .  \label{eq:estimate_error}
\end{eqnarray}



\section{Mutual Information \& Continuous Signals}
In this section we shall extend the theory of Mitter and Newton in information flow due to filtering to the setting of Markov diffusions. Certain Fisher information quantities will emerge as for canonical, and a crucial role is played by a Theorem of Mayer-Wolf and Zakai \cite{MWZ}.

\subsection{The Mayer-Wolf \& Zakai Theorem}

We shall now calculate the mutual information $\mathsf{I}\left( X\left(
t\right) ;Y_{0}^{t}\right) $ shared between the signal $X\left( t\right) $
and the measurements up to time $t$. This may be written as 
\begin{eqnarray}
\mathsf{I}\left( X\left( t\right) ;Y_{0}^{t}\right) =\mathbb{E}\left[ \ln 
\frac{\widehat{\rho} _{t}\left( X\left( t\right) \right) }{\rho_{t}\left(
X\left( t\right) \right) }\right]  \label{eq:MWZ1}
\end{eqnarray}
where $\widehat{\rho} _{t}\left( x\right) $ is the (random) conditioned density
function and $\rho_{t}\left(
x\right) $ is the average density function satisfying the Fokker-Planck
equation. Note that in (\ref{eq:MWZ1}), both of these are evaluated at $%
x=X\left( t\right) $. Note that $\widehat{\rho} _{t}\left( X\left( t\right)
\right) $ is the random variable $\omega \mapsto \widehat{\rho} _{t}\left(
X\left( t,\omega \right) ,\omega \right) $.

It turns out to be simpler to use Zakai's non-normalized density $\zeta
_{t}\left( x,\omega \right) $ as this satisfies a linear SDE. We have $\widehat{%
\rho} _{t}\left( x\right) =\frac{\zeta _{t}\left( x\right) }{\sigma
_{t}\left( 1\right) }$ so that 
\begin{eqnarray}
\mathsf{I}\left( X\left( t\right) ;Y_{0}^{t}\right) &=&\mathbb{E}\left[ \ln 
\frac{\zeta _{t}\left( X\left( t\right) \right) }{\rho_{t}\left( X\left(
t\right) \right) }\right] -\mathbb{E}\left[ \ln \sigma _{t}\left( 1\right) %
\right]  \notag \\
&=&\mathbb{E}\left[ \ln \frac{\zeta _{t}\left( X\left( t\right) \right) }{%
\rho_{t}\left( X\left( t\right) \right) }\right] -\frac{1}{2}\int_{0}^{t}%
\mathbb{E}\left[ | \pi _{t}\left( h\right) | ^{2}\right] d\tau .
\label{eq:MWZ2}
\end{eqnarray}

\bigskip

\begin{definition}
The \textbf{a-priori (or unconditional) Fisher information matrix}, $\mathsf{J}_{t}^{\rho }
$, is defined to be
\begin{eqnarray}
\mathsf{J}_{t}^{\rho}=\mathbb{E}\left[ \left. \left( \nabla \ln \rho_{t}\right) \Sigma \left( \nabla \ln \rho_{t}\right) ^{\top
}\right| _{x=X\left( t\right) }\right] .
\end{eqnarray}
The \textbf{a-posteriori (or conditional) Fisher information matrix based on the measurements}, $%
Y_{0}^{t}$, $\mathsf{J}_{t}^{\pi }$, is defined to be
\begin{eqnarray}
\mathsf{J}_{t}^{\pi }=\mathbb{E}\left[ \left. \left( \nabla \ln \widehat{\rho}
_{t}\right) \Sigma \left( \nabla \ln \widehat{\rho} _{t}\right) ^{\top }\right|
_{x=X\left( t\right) }\right] .
\end{eqnarray}
\end{definition}

Note that logarithmic derivatives allow us to use the unnormalized Zakai filter instead:
\begin{eqnarray}
\mathsf{J}_{t}^{\pi }\equiv \mathbb{E}\left[ \left. \left( \nabla \ln \zeta
_{t}\right) \Sigma \left( \nabla \ln \zeta _{t}\right) ^{\top }\right|
_{x=X\left( t\right) }\right] .
\end{eqnarray}

\begin{theorem}[Mayer-Wolf and Zakai]
The rate of change of the mutual information is
\begin{eqnarray}
\frac{d}{dt}\mathsf{I}\left( X\left( t\right) ;Y_{0}^{t}\right) =\frac{1}{2}%
\mathbb{E}\left[  | \varepsilon_t (h)  | ^{2}\right] -\frac{1}{2} \text{tr} \,
\left( \mathsf{J}_{t}^{\pi }-\mathsf{J}_{t}^{\rho}\right) .
\label{eq:MayerWolf_Zakai}
\end{eqnarray}
\end{theorem}

\subsection{Information Flow}

We remark that the equation (\ref{eq:MayerWolf_Zakai}) from the Mayer-Wolf
Zakai Theorem may be rewritten as 
\begin{eqnarray}
\frac{d}{dt}\mathsf{I}\left( X\left( t\right) ;Y_{0}^{t}\right) = \frac{d}{dt%
}\mathsf{I}\left( X_0^t ;Y_{0}^{t}\right) -\frac{1}{2} \text{tr} \, \left( 
\mathsf{J}_{t}^{\pi }-\mathsf{J}_{t}^{\rho}\right) .
\end{eqnarray}
where we use Duncan's Theorem.

From this, we see that the appropriate definitions generalizing Mitter and Newton \cite{MN05} are the following.

\begin{definition}
The \textbf{information supplied to the filter's memory storage} up to
time $t$ is defined to be
\begin{eqnarray}
\mathsf{S} (t) \triangleq \mathsf{I}\left( X_0^t;Y_{0}^{t}\right) +\mathsf{I}%
\left( X(0) ;Y(0) \right)
\end{eqnarray}
and the \textbf{information dissipated by the filter's memory storage} up to
time $t$ is 
\begin{eqnarray}
\mathsf{D} (t) \triangleq \frac{1}{2} \int_0^t \text{tr} \, \left( \mathsf{J}%
_{\tau}^{\pi }-\mathsf{J}_{\tau }^{\rho}\right) d \tau .
\label{eq:Dissip_Info_filter}
\end{eqnarray}
\end{definition}

We then have the decomposition
\begin{eqnarray}
\frac{d}{dt}\mathsf{I}\left( X\left( t\right) ;Y_{0}^{t}\right) = \dot{%
\mathsf{S}} (t)- \dot{\mathsf{D} } (t) .
\end{eqnarray}

\begin{example}
\label{example:KB_rates}
It is of interest to recover now the results of Mitter and Newton \cite{MN05} for the Kalman-Bucy filter. Evidently,
\begin{eqnarray}
\frac{d}{dt}\mathsf{S}_{\text{K-B}} (t) = \frac{1}{2} \text{tr} \, \big\{ C \widehat{V} (t) C^\top \big\}
\end{eqnarray}
from (\ref{eq:I_KB}). The unconditioned probability density is
$\rho_t$ is Gaussian of mean $\mu_t$ and variance $V(t)$, see (\ref{eq:mu_t}) an (\ref{eq:rate_V}), so that,
\begin{eqnarray*}
\nabla \rho_t (x) \equiv - V(t)^{-1}  \big( x - \mu_t \big)
\end{eqnarray*}
and therefore
\begin{eqnarray}
\text{tr} \, \mathsf{J}_{\tau }^{\rho} &\equiv& \mathbb{E} \left[ (\nabla \rho_t )^\top \Sigma  (\nabla \rho_t ) \right] 
\nonumber \\
&=& \mathbb{E} \left[ (x - \mu_t )^\top V(t)^{-1} \Sigma  V(t)^{-1} (x - \mu_t  ) \right] \nonumber \\
&=& \text{tr} \big\{ \Sigma V(t)^{-1} \big\} .
\end{eqnarray}
Similarly,  the conditioned probability density,
$\widehat{\rho}_t ( \cdot ,  \omega )$, is Gaussian of mean $\widehat{X}_t (\omega )$ and variance $\widehat{V}(t)$, so that,
\begin{eqnarray*}
\nabla \widehat{\rho}_t (x, \omega ) \equiv - \widehat{V}(t)^{-1}  \big( x - \widehat{X}_t(\omega )  \big)
\end{eqnarray*}
and
\begin{eqnarray}
\text{tr} \, \mathsf{J}_{\tau }^{\pi} &\equiv& \mathbb{E} \left[ (\nabla \widehat{\rho} _t )^\top \Sigma  (\nabla \widehat{\rho} _t ) \right] 
\nonumber \\
&=& \mathbb{E} \left[ (\widehat{X}_t )^\top \widehat{V}(t)^{-1} \Sigma  \widehat{V}(t)^{-1} (\widehat{X}_t) \right] \nonumber \\
&=& \text{tr} \big\{ \Sigma \widehat{V}(t)^{-1} \big\} .
\end{eqnarray}

Therefore, for the Kalman-Bucy filter, the dissipated information is
\begin{eqnarray}
 \mathsf{D}_{\text{K-B}} (t) = \frac{1}{2} \int_0^t \text{tr} \big\{ \Sigma \big( \widehat{V}(t)^{-1} -V(t)^{-1} \big) \big\}.
\end{eqnarray}
\end{example}

\bigskip

Let us now recall Lemma \ref{thm:dot_H} on the entropy production of the
unobserved process $X(t)$, that is, where we calculated the rate of change
of $\mathsf{H}_t = \mathsf{H} (X(t))$ to be 
\begin{eqnarray}
\frac{d}{dt}\mathsf{H}\left( X\left( t\right) \right) &=&\frac{1}{2}\int
\rho_{t}\left( \nabla \ln \rho_{t}\right) ^{\top }\Sigma \left( \nabla \rho
_{t}\right) +\int \rho_{t}\left( \nabla .u\right)  \notag \\
&\equiv & \frac{1}{2} \text{tr\thinspace }\mathsf{J}_{t}^{\rho}+\mathbb{E}%
\left[ \left. \left( \nabla .u\right) \right| _{x=X\left( t\right) }\right] ,
\label{eq:H_rate_1}
\end{eqnarray}
see (\ref{eq:lem_id}) where we recognize the first term as the unconditioned
Fisher information. We now give the rate of change for the entropy
conditioned on the available measurements.

\begin{proposition}
The conditional information for the signal, $X\left( t\right) $, given the
measurements up to time $t$, $Y_{0}^{t}$, has the rate of change 
\begin{eqnarray}
\frac{d}{dt}\mathsf{H}\left( X\left( t\right) |Y_{0}^{t}\right) =
\frac{1}{2}\mathrm{tr }\, \mathsf{J}_{t}^{\pi }+\mathbb{E}\left[ \left. \left( \nabla
.u\right) \right| _{x=X\left( t\right) }\right] 
-\frac{1}{2}%
\mathbb{E}\left[ |\varepsilon_t (h) |^{2}\right].  
\label{eq:rate_I_cond}
\end{eqnarray}
\end{proposition}

\begin{proof}
From equation (\ref{eq:mutual_cond}) we have
\begin{eqnarray*}
\frac{d}{dt}\mathsf{H}\left( X\left( t\right) |Y_{0}^{t}\right) =\frac{d}{dt}%
\mathsf{H}\left( X\left( t\right) \right) -\frac{d}{dt}\mathsf{I}\left(
X\left( t\right) ;Y_{0}^{t}\right) .
\end{eqnarray*}
The result then follows immediately from combining (\ref{eq:H_rate_1}) with (\ref{eq:MayerWolf_Zakai}) from the
Mayer-Wolf Zakai Theorem.
\end{proof}

\subsection{Rate of Information Dissipated by the Filter’s Memory Storage}

We now give an alternative form of the rate of information dissipated.
First, we need the following Lemma.

\begin{lemma}
\label{lem:Ddot}
Let $F:\Gamma \times \mathscr{C}\left[ 0,t\right] \mapsto \mathbb{R}$ be
jointly measurable, then
\begin{eqnarray*}
\mathbb{E}\left[ \frac{1}{\widehat{\rho}_{t}\left( x\right) }F\left(
x,Y_{0}^{t}\right) \right] =\int_{\Gamma }\mathbb{E}\left[ F\left(
x,Y_{0}^{t}\right) \right] dx.
\end{eqnarray*}
\end{lemma}

\begin{proof}
We have
\begin{eqnarray*}
\mathbb{E}\left[ \frac{1}{\widehat{\rho}_{t}\left( x\right) }F\left(
x,Y_{0}^{t}\right) \right]  &=&\int \frac{1}{\rho _{X\left( t\right)
|Y_{0}^{t}}\left( x|y_{0}^{t}\right) }F\left( x,y_{0}^{t}\right) 
\mathbb{P}
_{X\left( t\right) ,Y_{0}^{t}}\left( dx,dy_{0}^{t}\right)   \\
&=&\int F\left( x,y_{0}^{t}\right) \mathbb{P} _{Y_{0}^{t}}\left( dy_{0}^{t}\right)
\,dx   \\
&=&\int \mathbb{E}\left[ F\left( x,Y_{0}^{t}\right) \right] dx.
\end{eqnarray*}
\end{proof}

\bigskip

\begin{theorem}
The rate of change of the dissipation of information stored in the filter memory takes the
form
\begin{eqnarray}
\dot{\mathsf{D}} \left( t\right) =\frac{1}{2}\mathbb{E}\left[ \left. \left(
\nabla \ln \frac{\widehat{\rho}_{t}}{\rho _{t}}\right) ^{\top }\Sigma \left(
\nabla \ln \frac{\widehat{\rho}_{t}}{\rho _{t}}\right) \right| _{x=X\left(
t\right) }\right] .  \label{eq:dotD}
\end{eqnarray}
\end{theorem}

\begin{proof}
Writing $\nabla \ln \frac{\widehat{\rho}_{t}}{\rho _{t}}=\nabla \ln \widehat{\rho}%
_{t}-\nabla \ln \rho _{t}$ and expanding leads to the following expression
for the right hand side of (\ref{eq:dotD}):
\begin{eqnarray*}
\frac{1}{2}\text{tr}\mathsf{J}_{t}^{\pi }+\frac{1}{2}\text{tr}\mathsf{J}%
_{t}^{\rho }-\mathbb{E}\left[ \left. \left( \nabla \ln \widehat{\rho}_{t}\right)
^{\top }\Sigma \left( \nabla \ln \rho _{t}\right) \right| _{x=X\left(
t\right) }\right] 
\end{eqnarray*}
however, the cross-terms may be written as
\begin{eqnarray*}
-\mathbb{E}\left[ \frac{1}{\widehat{\rho}_{t}}\left. \left( \nabla \widehat{\rho}%
_{t}\right) ^{\top }\Sigma \left( \nabla \ln \rho _{t}\right) \right|
_{x=X\left( t\right) }\right] 
\end{eqnarray*}
and, by Lemma \ref{lem:Ddot}, this equals to 
\begin{eqnarray*}
&-&\int \mathbb{E}\left[ \left( \nabla \widehat{\rho}_{t}\left( x\right) \right)
^{\top }\Sigma \left( x\right) \left( \nabla \ln \rho _{t}\left( x\right)
\right) \right] dx \nonumber\\
&=&-\int \mathbb{E}\left[ \left( \nabla \widehat{\rho}%
_{t}\left( x\right) \right) ^{\top }\right] \Sigma \left( x\right) \left(
\nabla \ln \rho _{t}\left( x\right) \right) dx \\
&=&-\int \left( \nabla \rho _{t}\left( x\right) \right) ^{\top }\Sigma
\left( x\right) \left( \nabla \ln \rho _{t}\left( x\right) \right) dx \\
&=&-\text{tr}\mathsf{J}_{t}^{\rho }.
\end{eqnarray*}
Therefore the right hand side of (\ref{eq:dotD}) simplifies to $\frac{1}{2}$%
tr$\mathsf{J}_{t}^{\pi }-\frac{1}{2}$tr$\mathsf{J}_{t}^{\rho }$ which is the
form of $\mathsf{\dot{D}}\left( t\right) $ given in (\ref{eq:Dissip_Info_filter}).
\end{proof}

\begin{corollary}
In the decomposition $\frac{d}{dt}\mathsf{I}\left( X\left( t\right) ;Y_{0}^{t}\right) =
\dot{ \mathsf{S}} (t)- \dot{ \mathsf{D} } (t)$, we have both $\dot{ \mathsf{S} } (t) \geq 0$ and $\dot{ \mathsf{D} } (t)\geq 0$.
\end{corollary}

This follows from Duncan's Theorem, which says that $\dot{\mathsf{S} }(t)
\equiv \frac{1}{2} \mathbb{E} [ | \varepsilon_t (h ) |^2 ]$, and from the
specific form obtained in equation (\ref{eq:dotD}).

\bigskip

We note that the rate of information dissipation may be expressed in terms
of an average of the $\Gamma $-operator: 
\begin{eqnarray}
\dot{\mathsf{D}}_{t}\equiv \frac{1}{2}\mathbb{E}\left[ \left. \Gamma _{%
\mathcal{L}}\left( \ln \frac{\widehat{\rho}_{t}}{\rho _{t}},\ln \frac{\widehat{\rho}%
_{t}}{\rho _{t}}\right) \right| _{X\left( t\right) }\right] .
\label{eq:rate3}
\end{eqnarray}

\section{Feedback}

The Mayer-Wolf and Zakai Theorem may also be generalized in several directions \cite{MWZ}. 
In particular, in place of (\ref{eq:filter_obs}),
we can take the observations to satisfy 
\begin{eqnarray}
dY (t) = h\left( X(t), Y(t)\right)
dt+dU (t) ,
\label{eq:controlled_h}
\end{eqnarray}
without otherwise changed the form of the equation (\ref{eq:MayerWolf_Zakai}), or the ensuing analysis.
The dynamical and observational noises may also be taken to be correlated, but this results is a slight modification to 
(\ref{eq:MayerWolf_Zakai}).

If we wish to consider Maxwell D\ae mon type problems, it is more natural however to consider feedback of the system. Specifically, we consider the controlled flow
\begin{eqnarray}
dX (t) = v\left( X(t), \beta (t)\right) dt + B(X(t)) \, dW(t) ,
\label{eq:control_dynamics}
\end{eqnarray}
which is the same as (\ref{eq:filter_dynamics}) except that the velocity depends on a control process $\{ \beta (t) :
t \geq 0 \}$ which we take to be adapted to the filtration $\mathcal{Y}$ of the observations. 
The flow is then described by
\begin{eqnarray}
df(X(t)) = \mathcal{L}_{\beta } f (x) \vert_{x=X(t), \beta = \beta (t)} dt + B(f)\vert_{x=X(t)} dW(t),
\end{eqnarray}
where the controlled generator is
\begin{eqnarray}
\mathcal{L}_{\beta } f (x)
= v^i (x, \beta ) f_{,i}(x) + \frac{1}{2} \Sigma^{ij} f_{,ij} (x) .
\end{eqnarray}
We remark that \begin{eqnarray}
\pi_t \big( \mathcal{L}_{\beta (t)} f \big) \equiv \pi_t \big( \mathcal{L}_{\beta } f \big) \vert_{\beta = \beta (t)},
\label{eq:pigtails}
\end{eqnarray}
since the control, $\beta (t)$, is adapted to the observation's filtration. (We derive the corresponding 
Kushner-Stratonovich equation  in Section \ref{app:A}.)
We note that the Fokker-Planck equation takes the same form as before, (\ref{eq:Fokker_Planck}), but where the generator now has the time-dependent drift velocity
\begin{eqnarray}
\overline{v}^i (x, t ) = \mathbb{E} \big[ v^i \big( x , \beta (t)  \big) \big] .
\end{eqnarray}
With these obvious replacements, the Mayer-Wolf Zakai Theorem remains unchanged for the controlled flow (\ref{eq:control_dynamics}),
with $\mathcal{Y}$-adapted control process $\{ \beta (t) : t \ge 0 \}$, and with controlled observations (\ref{eq:controlled_h}).
For completeness, we give a derivation of this for (\ref{eq:control_dynamics}) with the basic uncontrolled observations
in Section \ref{app:B}.

The equations for entropy production and information flow are therefore formally unchanged if we allow feedback to the system by means of a controlled flow dynamics governed by a control process $\beta (\cdot )$ to the observation's filtration. Note that the unconditioned dynamics describes the average flow under the action of the D\ae mon, as the flow velocity field is $\overline{v} (x,t)$.

\section{Discussion}
We have extended the analysis of rate of change of the information stored and dissipated by a filter for the class of Markov diffusion models. The various rates of changes, e.g., equations (\ref{eq:lem_id}), (\ref{eq:dot_F}), (\ref{eq:dotD}), or (\ref{eq:rate3}), take the form of an average of a quadratic term involving the co-metric of logarithmic derivatives. This form is due to a type of de Bruin identity relating the derivatives of entropy under Gaussian perturbations to a type of Fisher information. 

We have included feedback as a feature, allowing the filter to act as a Maxwell's d\ae mon. Rather than being an entity with the propensity to change the state of the system in some thermodynamic sense (as in the original formulation where the d\ae mon opens a gate to allow one atom to pass at a time), we consider a controlled flow where the policy applied is determined by the filter. In principle it should then be possible to develop optimal control policies for appropriate performance objectives.

We have restricted our attention to systems where the state space is $\mathbb{R}^n$, however, our results may be naturally extended to diffusions on manifolds under the usual technical assumptions derived by Schwartz and Meyer, and indeed it is in this setting that geometric concepts such as co-metric first appeared, see for instance \cite{Emery}.

It is worth remarking on possible quantum mechanical generalizations. The appropriate setting is the quantum stochastic calculus of Hudson and Parthasarathy \cite{HP} \cite{Par92}, and the quantum filtering theory of Belavkin \cite{Belavkin1}. One of the issues facing quantum mechanics is that it is not possible to give a joint probability distribution to non-commuting observables, and here we find that system observables in the Heisenberg picture do not generally commute with the output processes that we measure. While an input-state-output description is possible, the input noise does not commute with the output observables. Nevertheless, there is a non-demolition property to these models that ensures that observables of the system at any time $t$ commute with the output observables at any time up to $t$ - so quantum filtering and prediction is possible, but not smoothing. Consequently we may give clear meaning to the mutual information shared between the system at a given time $t$ and the observations up to that time. This means that there should exist a version of the Mayer-Wolf Zakai Theorem, but that Duncan's theorem is problematic in the quantum case.
The quantum filtering theory leads to a density matrix valued process $\hat \varrho_t$ which gives the conditioned quantum mechanical state of the system conditioned on the (essentially classical) observations up to that time - so that we work on a hybrid classical-quantum probability space. The mutual information needed is the quantum relative entropy of $\hat \varrho_t$ relative to the classical marginals (the observations distribution) and the quantum marginals (the unconditioned state of the system at time $t$): mathematically this a Holevo information (also known as a Holevo $\chi$
quantity) \cite{Holevo73}. These have been calculated for quadrature (diffusive) and photon-counting (Poissonian) models by Barchielli and Lupieri \cite{BL06,BL06a}. What is of interest here is that there is a natural analogue to the co-metric for quantum dynamical semigroups and this is the dissipation associated to Lindblad generators \cite{Lindblad}. We will develop these ideas in a future publication.

\section{Appendix}

\subsection{Derivation of the Filter}
\label{app:A}
There a several ways to derive the Kushner-Stratonovich equation. One way is the reference probability method to obtain the Zakai form. We describe briefly another, called the characteristic method, which allows us to handle dynamics flows that are controlled.
Here we make the ansatz that 
\begin{eqnarray*}
d\pi_t (f) \equiv \mathcal{G}_{t}dt+\mathcal{H}_{t}^\top dY (t)
\end{eqnarray*}
for some $\mathcal{G}_{t}$ and $\mathcal{H}_{t}$ adapted to $\mathcal{Y}$.

Now let $C (t)$ be any 2nd-order process adapted to $\mathcal{Y}$, then by the projective property of the conditional expectation we must have that the error $\epsilon_t (f) = f(X(t)) - \pi_t (f)$ is orthogonal to the subspace of $\mathcal{Y}_0^t$-measurable variables, and so
\begin{eqnarray}
\mathbb{E}\left[ \epsilon_t (f) \, C (t)\right] =0,
\label{eq:orthog}
\end{eqnarray}

One now takes arbitrary $L^{2}$-functions, $g(t)=[g_1 (t), \cdots , g_p (t)]^\top$, and set 
\begin{eqnarray*}
dC (t)=  C (t)\, g(t)^\top dY (t).
\end{eqnarray*}

We have $d\big( \epsilon_t (f) \,C (t)\big) = d \epsilon_t (f) \,C (t)+ \epsilon_t (f) \,dC (t)+d\epsilon_t (f) \,dC (t) $ by the It\={o} product rule, so (\ref{eq:orthog}) implies $\mathbb{E}\left[ d\left( \epsilon_t (f)C (t)\right) \right] = I+II+III=0$ where
\begin{eqnarray*}
I &=& \mathbb{E} \bigg[  d \epsilon_t (f) \,C (t) \bigg] \\
&=& \mathbb{E} \bigg[ \big[ \big( \mathcal{L}_{\beta (t)} f\big) (X(t))- \mathcal{G}_t\big] C(t)  \bigg]  dt +\mathbb{E} \bigg[  - C(t)
\mathcal{H}_t dY(t) \bigg] \\
&=& \mathbb{E} \bigg[ \big[ \big( \mathcal{L}_{\beta (t)} f\big) (X(t))- \mathcal{G}_t\big] C(t)  \bigg]  dt -\mathbb{E} \bigg[  C(t)
\mathcal{H}_t^\top  h(X(t)) \bigg] dt ;\\
II &=& \mathbb{E} \bigg[   \epsilon_t (f) \, dC (t) \bigg] =
g(t)^\top \mathbb{E} \bigg[   \epsilon_t (f) \, h(X (t)) \, C(t) \bigg] dt ;\\
III &=& \mathbb{E} \bigg[  d \epsilon_t (f) \,dC (t) \bigg] =
\mathbb{E} \bigg[ - \mathcal{H}_t^\top dY(t) \, C (t) g(t)^\top dY(t)\bigg] 
\\
&=& - g(t)^\top \mathbb{E} \bigg[  \mathcal{H}_t C (t) \bigg] dt .
\end{eqnarray*}
As the $g$'s were arbitrary, the identity $I+II+III=0$ can be split up into parts that contain factors of the $g$'s,
\begin{eqnarray*}
g(t)^\top \mathbb{E} \bigg[ \big\{ \epsilon_t (f) h(X(t)) - \mathcal{H}_t \big\} C(t) \bigg] =0,
\end{eqnarray*}
and the rest which do not,
\begin{eqnarray*}
\mathbb{E} \bigg[ \big\{ (\mathcal{L}_{\beta (t)} f ) X(t)) -\mathcal{G}_{t} - \mathcal{H}_t^\top \pi_t (h) \big\} C(t) \bigg] =0.
\end{eqnarray*}
From these we have, respectively,
\begin{eqnarray*}
0 &=&  \mathbb{E} \big[   \epsilon_t (f) h(X(t)) - \mathcal{H}_t | \mathcal{Y}_0^t \big]  
= \pi_t(fh^\top ) - \pi_t (f) \pi_t (h ^\top ) -\mathcal{H}_t,
\end{eqnarray*}
and
\begin{eqnarray*}
0= \pi _t (\mathcal{L}_{\beta (t)} f )  -\mathcal{G}_{t} - \mathcal{H}_t^\top \pi_t (h)   .
\end{eqnarray*}

Therefore, $\mathcal{H}_t = \pi_t(fh^\top ) - \pi_t (f) \pi_t (h ^\top ) $, and 
$\mathcal{G}_{t} =\pi _t (\mathcal{L}_{\beta (t)}f )   - \mathcal{H}_t^\top \pi_t (h) $.
From this we deduce the form of the Kushner-Stratonovich filter.
The derivation here is basically the same as the standard uncontrolled case, though we note the equivalence (\ref{eq:pigtails}). The corresponding Zakai equation will be
\begin{eqnarray}
d\zeta _{t}(x, \omega ) =\mathcal{L}^\star_\beta \zeta _{t} (x, \omega ) \vert_{\beta = \beta (t, \omega )} \,dt+
\zeta_{t}(x, \omega ) \, h(x)^\top dY (t, \omega ).  \label{eq:Zakai_zo}
\end{eqnarray}

\subsection{Proof of the Mayer-Wolf Zakai Theorem with Feedback}
\label{app:B}
We derive the Mayer-Wolf Zakai with the additional feature of feedback to the dynamics as in (\ref{eq:control_dynamics}) with
$\beta (\cdot )$ assumed $\mathcal{Y}$-adapted.

From the SDE for $\zeta _{t}\left( x\right) $, equation (\ref{eq:Zakai_zo}), we
have
\begin{eqnarray*}
d\ln \zeta _{t}\left( x\right)  &=&\frac{1}{\zeta _{t}\left( x\right) }%
d\zeta _{t}\left( x\right) -\frac{1}{2}\frac{1}{\zeta _{t}\left( x\right)
^{2}}d\zeta _{t}\left( x\right) \,d\zeta _{t}\left( x\right)  \\
&=&\left[ \frac{1}{\zeta _{t}\left( x\right) }\left( \mathcal{L}_{\beta (t)}^{\star
}\zeta _{t}\right) \left( x\right) -\frac{1}{2}\left| h\left( x\right)
\right| ^{2}\right] dt+h\left( x\right) ^{\top }dY\left( t\right) .
\end{eqnarray*}
Noting that $X_{0}^{t}$ and $Y_{0}^{t}$ have zero cross-correlation, i.e., $%
dX\left( t\right) dY\left( t\right) ^{\top }\equiv 0$, we find
\begin{eqnarray*}
d\ln \zeta _{t}\left( X\left( t\right) \right)  &=&\left[ \frac{1}{\zeta _{t}%
}\left( \mathcal{L}_{\beta (t)}^{\star }\zeta _{t}\right) -\frac{1}{2}\left| h\right|
^{2}\right] _{x=X\left( t\right) }dt+h\left( X\left( t\right) \right) ^{\top
}dY\left( t\right)  \\
&&+\left. \left( \nabla \ln \zeta _{t}\left( x\right) \right) ^{\top
}\right| _{x=X\left( t\right) }dX\left( t\right)  \\
&&+\left. \frac{1}{2}\Sigma ^{ij}\left( x\right) \partial _{i}\partial
_{j}\ln \zeta _{t}\left( x\right) \right| _{x=X\left( t\right) }dt \end{eqnarray*}
\begin{eqnarray*}
&=&\bigg[ \frac{1}{\zeta _{t}}\left( \mathcal{L}_{\beta (t)}^{\star }\zeta _{t}\right) -\frac{1%
}{2}\left| h\right| ^{2} \\
&&+\frac{1}{2}\Sigma ^{ij}\left( x\right) \left( \frac{%
1}{\zeta _{t}\left( x\right) }\zeta _{t,ij}\left( x\right) -\frac{1}{\zeta
_{t}\left( x\right) ^{2}}\zeta _{t,i}\left( x\right) \zeta _{t,j}\left(
x\right) \right) \bigg]_{x=X\left( t\right) }dt \\
&&+h\left( X\left( t\right) \right) ^{\top }dY\left( t\right) +\frac{1}{%
\zeta _{t}\left( x\right) }\left. \left( \nabla \zeta _{t}\left( x\right)
\right) ^{\top }\right| _{x=X\left( t\right) }dX\left( t\right) .
\end{eqnarray*}
Taking the expectation leads to
\begin{eqnarray*}
\frac{d}{dt}\mathbb{E}\left[ \ln \zeta _{t}\left( X\left( t\right) \right) %
\right]  &=&\mathbb{E} \bigg[ \bigg( \frac{1}{\zeta _{t}\left( x\right) }\left( \mathcal{L}_{\beta (t)}^{\star }\zeta _{t}\right) \left( x\right) -\frac{1}{2}\left| h\left(
x\right) \right| ^{2} \\
&&+\frac{1}{2}\Sigma ^{ij}\left( x\right) \left( \frac{1}{\zeta _{t}\left(
x\right) }\zeta _{t,ij}\left( x\right) -\frac{1}{\zeta _{t}\left( x\right)
^{2}}\zeta _{t,i}\left( x\right) \zeta _{t,j}\left( x\right) \right)  \\
&&\left. +\left| h\left( x\right) \right| ^{2}+\frac{1}{\zeta _{t}}\left(
\nabla \zeta _{t}\left( x\right) \right) ^{\top }v(x)\bigg)\right| _{x=X\left(
t\right) } \bigg].
\end{eqnarray*}
Using the explicit form of the dual generator, we have
\begin{eqnarray*}
\frac{d}{dt}
\mathbb{E}\left[ \ln \zeta _{t}\left( X\left( t\right) \right) \right]  &=&%
\mathbb{E} \bigg[ \bigg( \frac{1}{\zeta _{t}\left( x\right) }\Sigma ^{ij}\left( x\right)
\zeta _{t,ij}\left( x\right) +\frac{1}{\zeta _{t}\left( x\right) }\Sigma
_{\; \, ,j}^{ij}\left( x\right) \zeta _{t,i}\left( x\right) \\
&&+\frac{1}{2}\Sigma
^{ij}_{\; \, ,ij} \left( x\right) 
-1\frac{1}{2}\Sigma^{ij}\frac{1}{\zeta_t^2}\zeta_{t,i}\zeta_{t,j} \\
&&\left.-\left( \nabla .v\right) \left( x, \beta (t) \right) +\frac{1}{2}\left| h\left(
x\right) \right| ^{2} \bigg) \right|_{x=X\left( t\right) } \bigg]
\end{eqnarray*}

A similar calculation leads to
\begin{eqnarray*}
\frac{d}{dt}\mathbb{E}\left[ \ln \rho_{t}\left( X\left( t\right)
\right) \right]  &=&\mathbb{E} \bigg[ \bigg( \frac{1}{\rho_{t}\left( x\right) }%
\Sigma ^{ij}\left( x\right) \rho_{t,ij}\left( x\right) +\frac{1}{\rho_{t}\left( x\right) }\Sigma _{\; \, ,j}^{ij}\left( x\right) 
\rho_{t,i}\left( x\right) \\
&&+\frac{1}{2}\Sigma ^{ij}_{\, \;,ij} \left( x\right) 
-\frac{1}{2} \Sigma^{ij}\frac{1}{\rho_t^2}\rho_{t,i}\rho_{t,j}\\
&&\left.-\left( \nabla .v\right) \left( x, \beta (t) \right) \bigg) \right|_{x=X\left( t\right) } \bigg]
\end{eqnarray*}
We therefore have
\begin{gather*}
\frac{d}{dt} \mathbb{E}\left[ \ln \frac{\zeta _{t}\left( X\left( t\right) \right) }{\rho_{t}\left( X\left( t\right) \right) }\right]  =\frac{1}{2}\mathbb{E}%
\left[ \left| h\left( X\left( t\right) \right) \right| ^{2}\right]  \\
-\frac{1}{2}\mathbb{E}\left[ \left. \Sigma ^{ij}\left( x\right) \left( 
\frac{1}{\zeta _{t}\left( x\right) ^{2}}\zeta _{t,i}\left( x\right) \zeta
_{t,j}\left( x\right) -\frac{1}{\rho_{t}\left( x\right) ^{2}}\rho_{t,i}\left( x\right) \rho_{t,j}\left( x\right) \right) \right|
_{x=X\left( t\right) }\right]  \\
+\mathbb{E}\left[ \left. \Sigma ^{ij}\left( x\right) \left( \frac{1}{\zeta
_{t}\left( x\right) }\zeta _{t,ij}\left( x\right) -\frac{1}{\rho_{t}\left( x\right) }\rho_{t,ij}\left( x\right) \right) \right|
_{x=X\left( t\right) }\right]  \\
+\mathbb{E}\left[ \left. \Sigma _{\;,i}^{ij}\left( x\right) \left( \frac{1%
}{\zeta _{t}\left( x\right) }\zeta _{t,j}\left( x\right) -\frac{1}{\rho_ {t}\left( x\right) }\rho_{t,j}\left( x\right) \right) \right|
_{x=X\left( t\right) }\right] .
\end{gather*}
The last two expectations however vanish identically. For instance,
\begin{eqnarray*}
\mathbb{E}\left[ \left. \Sigma ^{ij}\left( x\right) \left( \frac{1}{\zeta
_{t}\left( x\right) }\zeta _{t,ij}\left( x\right) \right) \right|
_{x=X\left( t\right) }\right]  &=&\mathbb{E}\left[ \left. \Sigma ^{ij}\left(
x\right) \left( \frac{1}{\hat{\rho} _{t}\left( x\right) }\hat{\rho} _{t,ij}\left(
x\right) \right) \right| _{x=X\left( t\right) }\right]  \\
&=&\int \Sigma ^{ij}\left( x\right) \mathbb{E}\left[ \hat{\rho} _{t,ij}\left(
x\right) \right] \;dx \\
&=&\int \Sigma ^{ij}\left( x\right) \rho_{t,ij}\left( x\right) \;dx \\
&\equiv &\mathbb{E}\left[ \left. \Sigma ^{ij}\left( x\right) \frac{1}{\rho_{t}\left( x\right) }\rho_{t,ij}\left( x\right) \right| _{x=X\left(
t\right) }\right] .
\end{eqnarray*}
The other identity,
\begin{eqnarray*}
 \mathbb{E}\left[ \left. \Sigma _{\;,i}^{ij}\left( x\right) \frac{1}{%
\zeta _{t}\left( x\right) }\zeta _{t,j}\left( x\right) \right| _{x=X\left(
t\right) }\right] \equiv \mathbb{E}\left[ \left. \Sigma _{\;,i}^{ij}\left(
x\right) \frac{1}{\rho_{t}\left( x\right) } \rho_{t,j}\left(
x\right) \right| _{x=X\left( t\right) }\right] ,
\end{eqnarray*}
being similarly established.
Using (\ref{eq:MWZ2}) we then have
\begin{eqnarray*}
\mathsf{I}\left( X\left( t\right) ;Y_{0}^{t}\right) =\frac{1}{2}\mathbb{E}%
\left[ \left| h\left( X\left( t\right) \right) \right| ^{2}\right] -\frac{1}{%
2}\mathbb{E}\left[  | \pi _{t}\left( h\right) | ^{2}\right] -\frac{1}{2}%
\left( \mathsf{J}_{t}^{\pi }-\mathsf{J}_{t}^{\rho}\right) .
\end{eqnarray*}
However, $\mathbb{E}\left[ \left| h\left( X\left( t\right) \right) \right|
^{2}\right] - \mathbb{E}\left[  | \pi _{t}\left( h\right) | ^{2}%
\right] =\mathbb{E}\left[  | h\left( X\left( t\right) \right) -\pi _{t}\left( h\right) | ^{2}\right] $, giving the desired result.
$\square$

\end{document}